\documentclass[11pt,a4paper]{article}

\usepackage[utf8]{inputenc}
\usepackage[T1]{fontenc}
\usepackage{microtype}

\usepackage[margin=1in]{geometry}
\usepackage{hyperref}       %
 \usepackage{booktabs}       %
 \usepackage{amsfonts}       %
 \usepackage{nicefrac}       %
 \usepackage{microtype}      %
 \usepackage{xcolor}         %
\usepackage{enumitem}

\usepackage{lipsum}

\usepackage{amsmath,mathtools}
\usepackage{amsthm,thmtools}

\usepackage{bm,dsfont}
\usepackage{csquotes}

\usepackage{derivative}

\usepackage{subcaption}
\usepackage{graphicx}
\usepackage{tikz}
\usepackage{pgfplots}
\pgfplotsset{compat=newest}

\allowdisplaybreaks

\usepackage[ruled,linesnumbered,noline]{algorithm2e}

\usepackage[skip=5pt plus 3pt minus 1pt,parfill=0pt]{parskip}
\makeatletter
\def\thm@space@setup{%
  \thm@preskip=\parskip \thm@postskip=0pt
}
\makeatother

\usepackage[
 backend=biber, style=numeric, citetracker, hyperref=true,
 maxcitenames=2, %
 sortcites,      %
 sorting=nyt,
 maxbibnames=12, %
 giveninits,     %
 date=year,      %
 isbn=false,     %
 url=false,      %
 doi=false,      %
 eprint=false,   %
]{biblatex}

\newbibmacro{string+doiurlisbn}[1]{
  \iffieldundef{doi}{
    \iffieldundef{url}{
      #1
    }{
      \href{\thefield{url}}{#1}
    }
  }{
    \href{http://dx.doi.org/\thefield{doi}}{#1}
  }
}

\DeclareFieldFormat
[article,inbook,incollection,inproceedings,patent,thesis,unpublished]
  {title}{\usebibmacro{string+doiurlisbn}{#1}}

\usepackage{hyperref}
\hypersetup{colorlinks,citecolor=green!60!black,linkcolor=blue!70!black,urlcolor=blue!70!black,breaklinks=true}
\usepackage[capitalize, noabbrev,nameinlink]{cleveref}
\usepackage{xurl}

\usepackage{xspace}

\definecolor{col1}{HTML}{56b3e9}
\definecolor{col2}{HTML}{e69f00}
\definecolor{col3}{HTML}{009e74}
\definecolor{col4}{HTML}{cc79a7}
\definecolor{col5}{HTML}{d55e00}
\definecolor{col6}{HTML}{0071b2}
\pgfplotsset{
    colormap={signalheat}{
        color(0cm)=(white);
        color(1cm)=(col1!70!white);
        color(2cm)=(col3!80!white);
        color(3cm)=(col2!85!white);
        color(4cm)=(col5)
    }
}

\declaretheorem{theorem}
\declaretheorem[sibling=theorem]{lemma}

\DeclarePairedDelimiter\ceil{\lceil}{\rceil}
\DeclarePairedDelimiter\floor{\lfloor}{\rfloor}

\newcommand{\cI}{\mathcal{I}}
\newcommand{\cH}{\mathcal{H}}
\newcommand{\pr}{\mathbb{P}}

\newcommand{\opt}{{\mathrm{OPT}}}
\newcommand{\optrand}{\mathrm{OPT}^{\mathrm{rand}}}
\newcommand{\optrandfull}{{\mathrm{OPT}^{\mathrm{rand, full}}_n}}
\newcommand{\optdet}{\mathrm{OPT}^{\mathrm{det}}}
\newcommand{\optdetfull}{{\mathrm{OPT}^{\mathrm{det, full}}_n}}
\newcommand{\pol}{{\mathcal{A}}}
\newcommand{\polmax}[1]{{\mathcal{A}(\max\{#1\})}}

\newcommand{\Beta}{\mathrm{Beta}}
\newcommand{\ILP}{\mathrm{ILP}}

\newcommand{\prob}[1]{\mathbb{P}\left[#1\right]}

\newcommand{\hist}{\mathcal{H}}
\newcommand{\success}{\text{success}}

\usepackage[colorinlistoftodos,prependcaption,textsize=scriptsize]{todonotes}
\title{The Secretary Problem with a Stochastic Precursor}

\author{Franziska Eberle\thanks{Institut für Mathematik, Technische Universität Berlin, Germany. Funded by the Deutsche Forschungsgemeinschaft (DFG, German Research Foundation) under Germany’s Excellence Strategy -- The Berlin Mathematics Research Center MATH$^+$ (EXC-2046/1, EXC-2046/2, project ID: 390685689).} \and Alexander Lindermayr\thanks{Institut für Mathematik, Technische Universität Berlin, Germany.}}
\date{}

\begin{document}

\maketitle

\begin{abstract}
In learning-augmented online algorithms, predictions 
are usually valued for what they say: a value estimate, a solution, or an algorithmic recommendation. 
This paper shows that predictions can also be valuable solely due to their arrival time.  
We study the fundamental secretary problem augmented with a stochastic 
precursor: a content-free signal that is guaranteed to arrive no later than 
the best item, but is otherwise stochastically timed. 
The signal does not carry any additional information; nevertheless, its 
timing alone changes the structure of optimal stopping. We characterize 
optimal policies in the random-order and adversarial-order models. In random 
order, a single uniformly timed precursor already gives success probability at 
least $\nicefrac12$, improving on the classic $\nicefrac1e$ benchmark. With increasingly 
late precursors, the success probability approaches $1$. In adversarial order, 
for which traditional models do not admit strong guarantees, 
sufficiently concentrated precursors recover constant success guarantees. Our 
results show that such novel forms of asynchronous temporal information are a 
distinct and powerful form of advice in online decision making
and may also be effective for other problems.
\end{abstract}

\section{Introduction}

The classic secretary problem is one of the most fundamental models of online decision making under uncertainty~\cite{dynkin1963optimum,Lindley1961DynamicPA,Gilbert1966RecognizingTM,Ferguson1989WhoST}.
Its parsimonious formulation has made it a natural testbed for richer information models.
Recent work on secretary and related problems studies algorithms with access to samples~\cite{KaplanNR25,CorreaCFOT25}, augmented by predictions of values, ranks, or thresholds~\cite{AntoniadisGKK23,FujiiY24,BraunS24,BalkanskiMM24,nourmohammadi2026ordinal,karisani2026secretary,BraunS24}, or relying on more general forms of advice~\cite{DuttingLLV24}.
A common feature of these models is that the additionally provided information is \emph{synchronous} and \emph{instructional}: it is available before the process starts (or arrives with the respective item) and directly informs when to stop. 

In many applications, however, side information is generated by a separate process and arrives on its \emph{asynchronous} timeline.
In hiring, for instance, a delayed recommendation or inference-model score may become available only after interviews have already started, and may indicate only that a particularly strong candidate is still among the remaining applicants.
Similar \emph{precursor events} arise in popularity forecasting through early attention measurements~\cite{SzaboH10,Tatar2014ASO,Hu2017PredictingKE}, in recommendation systems through evolving user-behavior signals~\cite{yang2022generalized,joulani13,Qian2024LearningWA}, and in scientific impact prediction or scouting through early indicators of future long-run success~\cite{DziallasB19,AbramoAF19,Adams2005EarlyCC}.
Such %
indicators may contain rich information in practice, but one basic aspect is purely temporal: they arrive at a time that only indicates
whether the most important event has already occurred.
This 
is much weaker than the predictions commonly used in learning-augmented algorithms since it does not %
inform a decision.

This observation motivates a simple question at the boundary of %
optimal stopping theory and learning-augmented online algorithms:
\emph{Can an online algorithm provably benefit from precursors whose only information is their timing?}
Our focus is on signals that are inherently \emph{asynchronous}: they may be produced by machine-learning pipelines, external observations, or side processes that do not align with the arrival sequence.
In this sense, our model is complementary to standard learning-augmented formulations.
Rather than predicting values, ranks, or thresholds, the signal reveals only the temporal cue that something important is still ahead.

In the standard secretary problem, %
$n$ secretaries (or \emph{items}) $\{1,\ldots,n\} =: [n]$ are presented online one by one, in adversarial or in uniformly random order, at times $1,\ldots,n$.
The goal is to \emph{maximize the probability} of stopping exactly when the unique best item arrives.
Without side information, the classic random-order problem has optimal success probability $\frac1e$~\cite{dynkin1963optimum,Lindley1961DynamicPA,Ferguson1989WhoST},
while adversarial order is much harsher: deterministic algorithms cannot guarantee positive success probability, and randomized algorithms achieve at most $\frac1n$.

We introduce an asynchronous and non-instructional precursor model for secretary problems.
After the arrival order has been fixed,
a time $S$ is drawn from a distribution over $[I]$, where $I$ denotes the arrival time of the unique best item.
When time $S$ is reached, the algorithm receives a content-free signal. %
This isolates the algorithmic value of temporal side information from richer advice.

The main message of the paper is that even this extremely weak signal model is surprisingly powerful; we give a high-level overview now and precise statements and definitions in \Cref{sec:our-contributions}:
\begin{itemize}
    \item \textbf{Full characterization of deterministic and randomized optimal policies} for $\alpha$-power distributed signal times under both arrival orders for every $\alpha > 0$; those policies always beat the respective benchmarks.
    \item[] \hspace{-\leftmargin} To highlight two concrete results, we show that, even without knowing $n$ in advance, in the 
    \item \textbf{random-order setting} a uniform signal suffices to succeed with probability $\frac12$, and in the 
	\item \textbf{adversarial-order setting} a uniform signal guarantees that even a deterministic policy succeeds with probability $\frac1n$. 
\end{itemize}

Conceptually, our results show that \emph{timing alone} can be algorithmically valuable.
The signal itself cannot be used to assess the quality of the current item directly. %
Nevertheless, it helps to outperform the classic benchmark, eliminates the need for knowing $n$, and,
when sufficiently concentrated near the optimum $I$, allows us to recover constant guarantees even under adversarial order.
This makes asynchronous temporal information a distinct and tractable resource and a novel augmented information model in online decision making.
We believe it will be useful beyond secretary problems.

\subsection{A more detailed overview of our results}\label{sec:our-contributions}
We focus on the $\alpha$-power signal model for $\alpha > 0$.
Conditioned on $I=i$, the signal time is distributed as $S = \ceil{i \cdot B}$ for $B \sim \Beta(\alpha,1)$; equivalently, each signal time $s \in [i]$ appears with probability
\[
	\pr[S=s \mid I=i]
	=
	\frac{s^\alpha-(s-1)^\alpha}{i^\alpha}.
\]
This Beta family gives a simple interpolation between signals that tend to occur early and signals that tend to occur late in the range $[i]$.
Two prominent special cases are $\alpha=1$, where $S$ is uniform on $[i]$,
and integer values $\alpha=m \geq 1$, where $S$ has the same distribution as the maximum of $m$ independent uniform signals in $[i]$. 
In general, for $\alpha \ll 1$, $S$ underestimates $i$, i.e., small $s$ are more likely, even for large~$i$, and for $\alpha \gg 1$, $S$ is likely close to~$i$; cf. \Cref{fig:beta} for representative densities.

We call a time/item a \emph{record} if the arriving item is the best item so far.
The \emph{threshold policy} $\pol(k)$ with threshold $k$ rejects all items before time $k$ and then accepts the first record from that time onward. 

\paragraph{Random order: exact optimal policy (\Cref{sec:random-order-optimal}).}
We fully characterize the optimal online policy (cf.~\Cref{thm:random-order-tight}).
If $\alpha\ge 1$, then the optimal rule is $\pol(S)$ (which does not need to know~$n$): wait for the signal and then accept the first record.
Notably, already a uniform signal boosts the probability of the next record being the overall maximum above~$\frac12$.   
If $0<\alpha<1$, then the optimal policy is $\pol(\max\{S,k_n\})$ relying on~$n$, where $k_n/n \to (1-\alpha)^{1/\alpha}$.
Asymptotically (cf.  \Cref{fig:beta}),  the worst-case probability
\[
	\opt_n(\alpha) \xrightarrow{n \to \infty} \opt(\alpha)=
	\begin{cases}
		\displaystyle \frac{\alpha+(1-\alpha)^{1+1/\alpha}}{\alpha+1}, & \text{if } 0<\alpha<1,  \\[1.2ex]
		\displaystyle \frac{\alpha}{\alpha+1},                         & \text{if } \alpha\ge 1.
	\end{cases}
\]
In particular, $\opt(\alpha) > \frac1e$, beating the benchmark for all $\alpha>0$, and $\opt(\alpha) \xrightarrow{\alpha\to 0} \frac1e$.

\begin{figure}[t]
	\centering
    \begin{subfigure}[t]{0.49\textwidth}
		\centering
		\begin{tikzpicture}
			\begin{axis}[
					width=\linewidth,
					height=0.6\linewidth,
					xmin=0,
					xmax=1,
					ymin=0,
					ymax=5.2,
					samples=400,
					xlabel={$x$},
					ylabel={signal density},
					axis lines=left,
					grid=both,
					tick label style={font=\scriptsize},
					label style={font=\footnotesize},
					legend style={draw=none, fill=none, at={(0.5,1.07)}, anchor=north, font=\scriptsize},
				]
				\addplot[very thick, color=col2, domain=0.013:1] {0.1*x^(-0.9)};
				\addlegendentry{$\alpha=\makebox[1.3em][r]{0.1}$}

				\addplot[very thick, color=col2, dashed, domain=0.01:1] {0.5/sqrt(x)};
				\addlegendentry{$\alpha=\makebox[1.3em][r]{0.5}$}

				\addplot[very thick, color=black, domain=0:1] {1};
				\addlegendentry{$\alpha=\makebox[1.3em][r]{1}$}

				\addplot[very thick, color=col1, domain=0:1] {2*x};
				\addlegendentry{$\alpha=\makebox[1.3em][r]{2}$}

				\addplot[very thick, color=col1, dashed, domain=0:1] {10*x^9};
				\addlegendentry{$\alpha=\makebox[1.3em][r]{10}$}
			\end{axis}
		\end{tikzpicture}
	\end{subfigure}
	\begin{subfigure}[t]{0.49\textwidth}
		\centering
		\begin{tikzpicture}
			\begin{axis}[
					width=\linewidth,
					height=0.6\linewidth,
					xmin=0.0001,
					xmax=100,
					ymin=0.33,
					ymax=1.05,
					samples=200,
					xlabel={$\alpha$},
					ylabel={asymptotic success prob.},
					axis lines=left,
					grid=both,
					tick label style={font=\scriptsize},
					label style={font=\footnotesize},
					legend style={draw=none, fill=none, at={(0.98,0.02)}, anchor=south east, font=\scriptsize},
					ytick={0.3678794412,0.5,1},
					yticklabels={$1/e$,$1/2$,$1$},
					xtick={0.0001,0.2,0.4,0.6,0.8,1,10,100},
					xticklabels={$0$,,$0.4$,,$0.8$,$1$,$10$,$100$},
					x coord trafo/.code={\pgfmathparse{#1 <= 1 ? #1 : 1 + ln(#1)/ln(10)}},
					x coord inv trafo/.code={\pgfmathparse{#1 <= 1 ? #1 : 10^(#1-1)}},
				]
				\addplot[very thick, color=col2, domain=0.001:1] {(x + (1-x)^(1 + 1/x))/(x+1)};
				\addlegendentry{$\opt(\alpha)$, $\alpha < 1$}

				\addplot[very thick, color=col1, domain=1:100] {x/(x+1)};
                \addlegendentry{$\opt(\alpha)$, $\alpha > 1$}

				\addplot[black, dashed] coordinates {(1,0.33) (1,1.02)};
				\addplot[gray, dashed] coordinates {(0.001,0.3678794412) (100,0.3678794412)};
				\addplot[gray, dashed] coordinates {(0.001,1) (100,1)};
				\addplot[only marks, mark=*, black] coordinates {(1,0.5)};

			\end{axis}
		\end{tikzpicture}
	\end{subfigure}\hfill
	\caption{Left: PDFs of $\mathrm{Beta}(\alpha, 1)$. \quad Right: Plot of $\opt(\alpha)$ from \Cref{thm:random-order-tight}.}
    \label{fig:beta}
\end{figure}
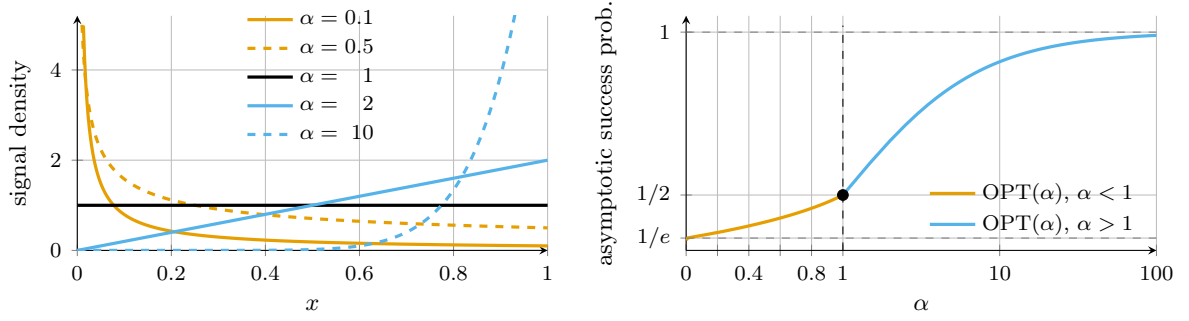

\paragraph{Random order: robustness (\Cref{sec:random-order-robustness}).} 
Knowing $\alpha\notin (0,1)$ is sufficient to recover the optimal success probability achievable with precise knowledge of~$\alpha$. 
In general, the qualitative structure is robust: 
If the policy only knows a conservative estimate $\hat\alpha \le \alpha$ of the true parameter~$\alpha$, a threshold policy tuned to $\hat\alpha$ still guarantees $\opt(\hat\alpha)> \frac1e$ asymptotically (cf.~\Cref{thm:conservative-misspecification}); underestimating~$\alpha$ is safe. 
We also give an explicit formula for the asymptotic success probability for arbitrary $\alpha,\hat \alpha > 0$:
\[
		g(\alpha,\hat\alpha)
		=
		\begin{cases}
			\displaystyle
			\frac{\alpha}{\alpha+1}
			+
			\frac{1-\alpha}{\alpha}(1-\hat\alpha)^{1/\hat\alpha}
			-
			\frac{(1-\hat\alpha)^{(\alpha+1)/\hat\alpha}}{\alpha(\alpha+1)},
			 & \text{if } 0<\hat\alpha\le 1, \\[3ex]
			\displaystyle \frac{\alpha}{\alpha+1},
			 & \text{if } \hat\alpha\ge 1.
		\end{cases}
\]

\paragraph{Adversarial order: deterministic and randomized policies (\Cref{sec:adversarial}).}
We also study an adversarial-order model with an $\alpha$-power signal, where $i^\star$ is some (adversarially) fixed best item. %
For randomized policies, we give the exact worst-case success probability
\[
	\opt^{\mathrm{rand}}_{n}(\alpha)=\frac{n^\alpha}{\sum_{j=1}^n j^\alpha} \ ,
\]
achieved by %
$\polmax{S,R}$, where $R$ is a (non-trivial) random threshold that depends on $n$ (cf.\ \Cref{thm:adversarial-rand-opt}).
This yields sharp asymptotic regimes: if $\alpha=o(n)$, then $\opt^{\mathrm{rand}}_{n}(\alpha) = o(1)$; if $\frac\alpha n\to c\in(0,\infty)$, then $\opt^{\mathrm{rand}}_{n}(\alpha)\to 1-e^{-c}$ (cf.\ \Cref{cor:adversarial-rand-regimes}).
For deterministic policies, we show that $\pol(S)$
is optimal for every $\alpha>0$ (cf.\ \Cref{thm:adversial-det}), with exact worst-case success probability 
\[
	\opt^{\mathrm{det}}_{n}(\alpha)=1-\left(1-\frac1n\right)^\alpha \ .
\]
If $\alpha=o(n)$, then $\opt^{\mathrm{det}}_{n}(\alpha)=o(1)$, while $\frac\alpha n\to c\in(0,\infty)$ implies $\opt^{\mathrm{det}}_{n}(\alpha) \to 1-e^{-c}$ (cf.\ \Cref{cor:adversarial-det-regimes}).
Thus, qualitatively, deterministic and randomized policies exhibit the same asymptotic behaviors; randomization improves the finite-$n$ guarantees in the small-$\alpha$ regime. %

\paragraph{Adversarial order: full signal histories (\Cref{sec:multiple-vs-alpha}).}
When $\alpha=m$ is an integer, the $\alpha$-power signal %
can be interpreted as the last of $m$ independent uniform signals in $[i^\star]$.
We therefore also study the richer \emph{full-history} model in which the algorithm observes all $m$ signal times, rather than only their maximum.
For randomized policies, we show that %
this gives no additional worst-case power (cf.\ \Cref{lem:last-signal-rand}). 
For deterministic policies, seeing the full history strictly improves the guarantee, already for $n=4$ and $m=2$.
To capture this additional power, we characterize the deterministic full-history optimum by an integer linear program (\Cref{thm:det-ilp}). %
For $m=2$ we simplify the characterization of optimal solutions (\Cref{thm:m2-det}).
Thus, unlike in the randomized case, the internal history of asynchronous signals can be algorithmically meaningful for deterministic adversarial-order stopping. %

\paragraph{Experimental results (\Cref{sec:experiments}).}
Our experimental results complement the theory with simulations of the actual online policies.
In the random-order model, we show that the theorized gains over the classic secretary benchmark are visible already at moderate problem sizes, and that the threshold policies behave robustly when the signal-quality parameter is conservatively misspecified.
We further test corrupted asynchronous signals 
and show that a simple fall-back-to-classic rule provides a smooth interpolation between the learned and prediction-free regimes.
In the adversarial-order model, we confirm the behavior in the $\alpha=cn$ regime predicted by the analysis, and demonstrate that full signal histories can strictly help deterministic algorithms.

\paragraph{Conclusion.}
Taken together, %
our results isolate a new source of algorithmic power in online optimal stopping: \emph{not predictive content, but predictive timing}.
Asynchronous signals help even in the notoriously hard adversarial regime, and in random order they beat the classic optimum already with a single uniform signal.
Most proofs are deferred to the appendix.

\subsection{Further related work}

A natural extension of the classic secretary problem is the \emph{full information} setting, where item values are drawn i.i.d.\ from a known distribution \cite{Gilbert1966RecognizingTM,EsfandiariHLM20,Nuti22}. 
More broadly, optimal stopping problems such as prophet inequalities~\cite{CorreaFHOV18,Lucier17,HillK92}, Pandora's box problems~\cite{BeyhaghiC23,Weitzman1978}, and variants there-of typically assume full distributional knowledge. 
A recent line of work relaxes full distributional knowledge by providing the algorithm with \emph{samples} from the underlying distribution(s), both in the prophet \cite{AzarKW19,CorreaC0S22,CorreaDFS22,RubinsteinWW20} and secretary \cite{KaplanNR25,DuttingLLV24,CorreaCFOT25} settings. 
Another way to relax full distributional knowledge is by only assuming \emph{distributional advice} such as partial or approximate knowledge of the input distribution \cite{DiakonikolasKTV21,CuiDinitz26,DinitzILMNV24,ABD024,MoseleyNPZ25}. 
While these models and our model of knowing the signal distribution may appear similar, the distinction is structural: 
All of these models provide information about the (distribution of the) item values, whereas the precursor signal only carries \emph{temporal} information dependent on the arrival time of the optimum and cannot be used to evaluate any candidate directly.

In  the \emph{learning-augmented algorithms} framework \cite{MitzenmacherV22,alps}, algorithms are equipped with a (potentially erroneous) machine-learned prediction, and their performance guarantee is typically analyzed as a function of some prediction error. 
This paradigm has been investigated for a variety of problems, e.g.\ caching 
\cite{LykourisV21,Wei20,BansalCKPV22,AntoniadisCEPS23,EKMM24}, scheduling \cite{LattanziLMV20,AzarLT21,LindermayrM25}, matching and allocation \cite{JinM22,SpaehE23,ChooGL024}, submodular maximization \cite{AgarwalB24}, and online learning \cite{KhodakBTV22,RamanT24}. 
In the optimal stopping context, predictions typically take the form of the (final) rank or value of the current item and, if perfect, directly allow to stop at the maximum \cite{AntoniadisGKK23,FujiiY24,BraunS24,BalkanskiMM24,nourmohammadi2026ordinal,karisani2026secretary,BraunS24}. 
In contrast, the precursor only signals that the optimum has not yet passed, and the question is how to optimally exploit a probabilistic but trustworthy temporal clue.

Closest in spirit to our model is recent work 
on non-clairvoyant scheduling \cite{GuptaKLSY25,BCLS25}, 
where a delayed external signal provides partial information about each job's 
characteristics. 
Our precursor signal shares the flavor of delayed and reliable side information, 
but differs in that it is tied to the optimality of an item rather than to  intrinsic features such as processing times.

\section{Random-order model: optimal policies}\label{sec:random-order-optimal}

We start with the random-order model. 
Here, the arrival time $I$ of the best item is uniformly distributed over all times $[n]$, that is, $\pr [I=i] = \frac1n$ for each $i \in [n]$. 
The baseline is the classic threshold policy~$\pol\big(\frac ne\big)$, 
which achieves the optimal success probability $\frac1e \approx 0.37$~\cite{dynkin1963optimum,Lindley1961DynamicPA,Ferguson1989WhoST}.

\subsection{Warm-up: special cases} %

\paragraph{Single uniform signal.} 
We show that we can improve substantially over $\frac1e$ 
even with a single uniformly distributed signal $S \sim \mathcal U([I])$ 
(equivalent to the $\alpha$-power model with parameter $\alpha=1$).

We consider the policy $\pol(S)$, that is, we reject all items before time~$S$ and accept the first record in~$\{S,\ldots,n\}$. 
For the analysis of the success probability of this policy, 
fix a realization~$I = i \in [n]$ and~$S = s \in [i]$. 
If~$i = s$, then~$\pol(S)$ picks~$i$ and succeeds. %
If~$s < i$ (and hence $i \geq 2$), then~$\pol(S)$ wins if the best item among %
$[i-1]$ arrives before~$s$. This makes~$s$ a non-record, lets $\pol(S)$ stop at $i$ and happens with probability $\frac{s-1}{i-1}$. Averaging over~$s \in [i]$ yields success probability
\(
	\frac{1}{i} \sum_{s=1}^i \frac{s-1}{i-1} = \frac12.
\)
Finally, averaging over $i \in [n]$ gives an overall success
probability of $\frac{1}{n} + \frac{n-1}{2n} = \frac{n+1}{2n} > \frac{1}{2}$.

\paragraph{$\alpha$-power signals.}

As a second warm-up, we extend the preceding argument to an $\alpha$-power signal
with an arbitrary fixed parameter $\alpha>0$.
We again use the policy~$\pol(S)$: reject all items before the signal time~$S$
and accept the first record from time~$S$ onward.

For the analysis, condition on $I=i$.
The case $i=1$ is trivial as then $S=1$ and the policy succeeds.
Hence, suppose $i\geq 2$, and condition further on $S=s\in[i]$.
If $s=i$, then the signal occurs at the best item and $\pol(S)$ accepts it.
If $s<i$, then $\pol(S)$ succeeds exactly when the best item among the first
$i-1$ times appears before time~$s$, which happens with probability $\frac{s-1}{i-1}$.
Therefore,
\begin{align*}
	\pr[\success \mid I = i] &= \frac{i^\alpha - (i-1)^\alpha}{i^\alpha} + \sum_{s=1}^{i-1} \frac{s-1}{i-1} \cdot \frac{s^\alpha - (s-1)^\alpha}{i^\alpha} \\
	&=  1 - \frac{(i-1)^\alpha}{i^\alpha} + \frac{(i-2)(i-1)^\alpha}{(i-1)i^\alpha} - \sum_{s=1}^{i-2} \frac{s^\alpha}{(i-1)i^\alpha}	\\ %
	&= 1 - \sum_{s=1}^{i-1} \frac{s^\alpha}{(i-1)i^\alpha} 
    = 1 - \left( \frac{1}{\alpha+1} \frac{(i-1)^\alpha}{i^\alpha} + O\!\left(\frac1i \right) \right)
	= \frac{\alpha}{\alpha+1} +  O\!\left(\frac1i \right).
    \shortintertext{Averaging over all $i \in [n]$ gives a total success probability of at least}
    \pr[\success]
	& = \frac{1}{n} \sum_{i=1}^n \pr[\success \mid I = i]
	= \frac{\alpha}{\alpha+1} + O\!\left( \frac{\log n}{n} \right)
	\xrightarrow{n \to \infty} \frac{\alpha}{\alpha+1}.
\end{align*}

The limiting success probability is increasing in $\alpha$ and tends to $1$ as
$\alpha\to\infty$, so increasingly late signals make $\pol(S)$ nearly perfect.
In the remainder, we show that, for all $\alpha\geq 1$,
$\pol(S)$ is indeed optimal. %
For $\alpha\to 0$, the signal concentrates at $S=1$, and the
asymptotic success probability of $\pol(S)$ vanishes.
In that regime, the signal alone should not be used as the threshold: the optimal
policy combines $S$ with an $n$-dependent deterministic threshold, achieving a
success probability strictly larger than $\frac1e$ for every $0<\alpha<1$ and
approaching the classic $\frac1e$-guarantee as $\alpha\to 0$.

\subsection{Bellman recursion and optimal policies}

We first isolate the dynamic-programming structure of the problem in the following lemma.
The \emph{history} at time~$t$ describes the observed relative ranks as well as the signal time if $S \leq t$. 

\begin{restatable}[Bellman recursion]{lemma}{lemmaBellman}\label{lem:random-order-bellman}
	Fix $\alpha>0$.
	There is an optimal policy that never stops before the signal and only stops at record times.
   	Let $\Pi_t$ be the optimal success probability conditioned on not having stopped before time~$t$ and having observed a history in which
	\begin{enumerate}[label=(\roman*),nosep]
		\item the signal has already appeared, and
		\item time $t$ is a record.
	\end{enumerate}
	Then, $\Pi_t$ depends only on $t$.
    With the normalizing term $\Psi_t:=t^{1-\alpha}+\sum_{i=t+1}^n i^{-\alpha}$ and the (unnormalized) value $\Phi_t:=\Psi_t\Pi_t$,
	we have $\Phi_n=n^{1-\alpha}$ and, for every $t<n$,
	\begin{equation}\label{eq:bellman-random-order}
		\Phi_t=
		\max\left\{
		t^{1-\alpha},
		\sum_{u=t+1}^n \frac{t}{u(u-1)}\Phi_u
		\right\}.
	\end{equation}
\end{restatable}

The recursion is a posterior comparison. 
Suppose the signal has already arrived and the current time~$t$ is a record. 
Conditioned on this, the exact signal time no longer matters: 
we obtain the same likelihood for every possible location of the maximum.
In the unnormalized value~$\Phi_t=\Psi_t\Pi_t$, stopping at~$t$ gives~$t^{1-\alpha}$, the posterior weight that the current record is the overall maximum.
If we continue, the only relevant future times are future records.
The probability that the next record occurs at~$u>t$ is the standard record factor
\(
 \frac t{u(u-1)},
\)
and the value from that point on is~$\Phi_u$.
Thus, the Bellman recursion simply compares stopping now with waiting for the next record.

A simplified view of the recursion asks if one should stop at~$t$, assuming that any later records will imply immediate stopping. 
In this case, the continuation value becomes 
\(
    C_t := t\sum_{u=t+1}^n \frac{1}{(u-1)u^\alpha},
\)
whereas the stopping value is~$t^{1-\alpha}$. 
Hence, comparing $G_t := t^{1-\alpha} - C_t$ to $0$ governs the local decision. 
Exactly here the parameter~$\alpha$ comes into play: the larger~$\alpha$, the faster future weights~$u^{-\alpha}$ (and the future probabilities for records) decay and the more likely the current record is the maximum. 

For~$\alpha\ge1$, Bernoulli's inequality ensures that the continuation value telescopes:
\[
    C_t
    \le
    t\sum_{u=t+1}^n\bigl((u-1)^{-\alpha}-u^{-\alpha}\bigr)
    < t^{1-\alpha}.
\]
Thus,~$G_t>0$ for each~$t$, and stopping at the first record after the signal is better; $\pol(S)$ is optimal. 
We can evaluate the success probability and its asymptotic behavior as seen in the warm-up.

For~$0<\alpha<1$, $G_t$ switches signs from negative to positive exactly once. 
On the scale~$t=\floor*{cn}$,
\[
    n^{\alpha-1}G_t
    \longrightarrow
    c^{1-\alpha}-c\int_c^1 y^{-(\alpha+1)}\,dy
    =
    \frac{c-(1-\alpha)c^{1-\alpha}}{\alpha}\, .
\]
This limit is negative for~$c<(1-\alpha)^{1/\alpha}$ and positive for~$c>(1-\alpha)^{1/\alpha}$
and we can show that the sign change induces a deterministic cutoff $k_n$ with~$\frac{k_n}n\to(1-\alpha)^{1/\alpha}$: Records after the signal but before the cutoff point should still be skipped, while records after both the signal and the cutoff point should be accepted, making $\polmax{S,k_n}$ optimal. 
Again, we can compute the optimal success probability and derive the theorem below.

\begin{restatable}{theorem}{thmOptimalRandomOrder}\label{thm:random-order-tight}
	For $\alpha\ge 1$, the policy $\pol(S)$ %
    is optimal.
    For $0<\alpha<1$, there is a threshold $k_n\in[n]$ such that the policy $\pol(\max\{S,k_n\})$ is optimal and $\frac{k_n}{n} \xrightarrow{n \to \infty} (1-\alpha)^{1/\alpha}$. 
    If $\opt_n(\alpha)$ denotes the optimal success probability on instances of length~$n$, then 
	\[
		\opt_n(\alpha) \xrightarrow{n \to \infty} \opt(\alpha):=
		\begin{cases}
			\displaystyle \frac{\alpha+(1-\alpha)^{1+1/\alpha}}{\alpha+1}, & 0<\alpha<1,  \\[1.2ex]
			\displaystyle \frac{\alpha}{\alpha+1},                         & \alpha\ge 1 \, .
		\end{cases}
	\]
\end{restatable}

\section{Random-order model: robustness}\label{sec:random-order-robustness}

The threshold $k_n$ in \Cref{thm:random-order-tight} depends on the exact knowledge of $\alpha$.
In general, $\alpha$ might not be known. 
In order to ensure robustness to parameter misspecification, we also investigate the asymptotic behavior of $\polmax{S,\lceil\beta n\rceil}$ for~$\beta\in[0,1]$ on instances with parameter~$\alpha$. We show that \[
    \prob{\success} 
        \xrightarrow{n\to\infty}
        f(\alpha,\beta) 
        = 
        \frac{\alpha}{\alpha+1}
		+
		\frac{1-\alpha}{\alpha}\beta
		-
		\frac{\beta^{\alpha+1}}{\alpha(\alpha+1)}.
\] 
Letting $\beta^\star(\alpha)$ denote the maximizer for $\alpha > 0$, 
we recover $f(\alpha,\beta^\star(\alpha))  = \opt(\alpha)$. 
More importantly, for an estimate~$\hat\alpha$ of~$\alpha$, we still obtain good bounds using $\beta^\star(\hat\alpha)$ (instead of $\beta^\star(\alpha)$).

\begin{restatable}[Smoothness and Robustness]{theorem}{thmRobustness}\label{thm:conservative-misspecification}
	  Let $\alpha > 0$ and $\hat \alpha > 0$. %
	The policy $\pol(\max\{S,k_n\})$ with
	\(
		k_n=\ceil*{\beta^\star(\hat\alpha)n}%
	\)
	has an asymptotic success probability of at least
    \[
		g(\alpha,\hat\alpha)
		=
		\begin{cases}
			\displaystyle
			\frac{\alpha}{\alpha+1}
			+
			\frac{1-\alpha}{\alpha}(1-\hat\alpha)^{1/\hat\alpha}
			-
			\frac{(1-\hat\alpha)^{(\alpha+1)/\hat\alpha}}{\alpha(\alpha+1)},
			 & 0<\hat\alpha\le 1, \\[3ex]
			\displaystyle \frac{\alpha}{\alpha+1},
			 & \hat\alpha\ge 1.
		\end{cases}
	\]
    In particular, if $\hat\alpha\le \alpha$, then
	\(
		g(\alpha,\hat\alpha)\ge \opt(\hat\alpha) \ge \frac 1e.
	\)
\end{restatable}
We highlight that for $\alpha \notin (0,1)$, the optimal policy does not require a threshold besides~$S$, and hence, the optimal success probability can be achieved only knowing $\alpha\geq 1$.
Thus, the more demanding regime is $\alpha < 1$, in which our guarantee~$g(\alpha,\hat\alpha)$ degrades smoothly in $|\alpha - \hat\alpha|$; cf. \cref{fig:adversarial-order}.

\section{Adversarial order}\label{sec:adversarial}

\begin{figure}[t]
	\centering
	\begin{subfigure}[t]{0.49\textwidth}
		\centering
		\includegraphics[width=0.9\linewidth]{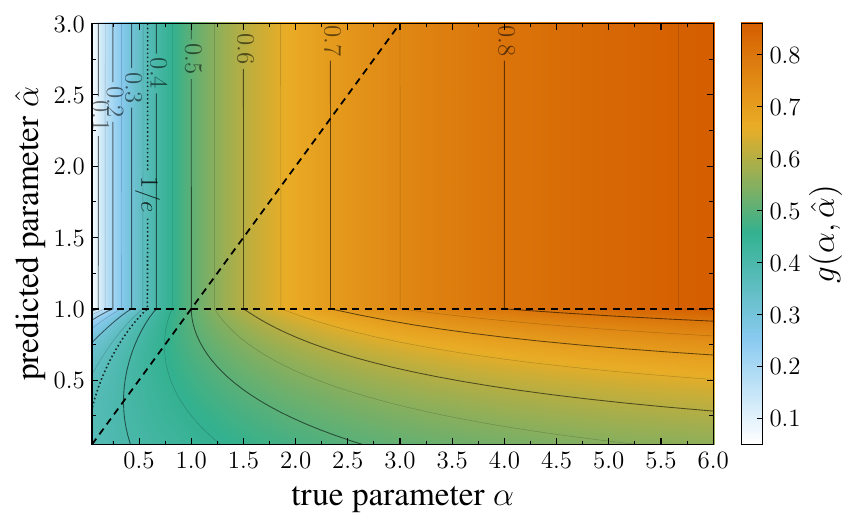}
	\end{subfigure}\hfill
	\begin{subfigure}[t]{0.49\textwidth}
		\centering
		\includegraphics[width=0.9\linewidth]{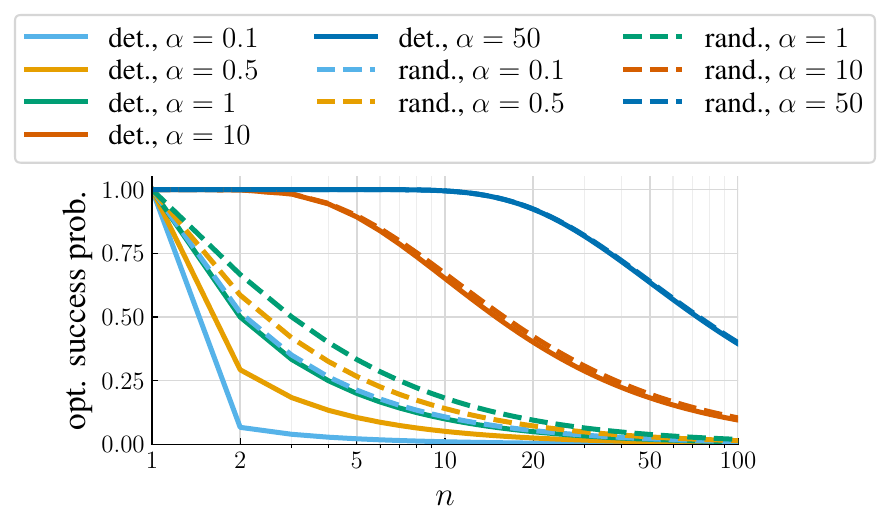}
	\end{subfigure}
	\caption{Left: Heatmap of the asymptotic guarantee $g(\alpha,\hat\alpha)$ from \Cref{thm:conservative-misspecification}. Right: Optimal success probability in the adversarial-order model. Solid lines show the det.\ opt.\ $1-(1-1/n)^\alpha$; dashed lines show the rand.\ opt.\ $n^\alpha/\sum_{j=1}^n j^\alpha$.}
	\label{fig:adversarial-order}
\end{figure}
In this section, we present our tight results for adversarial arrival orders and policies enhanced by a single $\alpha$-power signal.
Here, the unique maximum arrives at an adversarially chosen time $i^\star \in [n]$ 
and the policy receives an $\alpha$-power signal at time $S \in [i^\star]$.

To describe instances, instead of talking about the relative ranks of items, we will think about the item at time~$i$ as having a \emph{value} $v_i$. 
This allows us to represent an instance $\cI$ by its value vector~$v \in \mathbb N_0^n$.
In fact, all of our upper bounds on the optimal success probabilities in this section use the family~$\{\cI_i\}_{i=1}^n$ of \emph{hard} instances where $\mathcal I_i := (1,2,\ldots,i,0,\ldots,0) \in \mathbb N_0^n$
for $i \in [n]$.

\subsection{Randomized policies}

For randomized policies without information augmentation,
the best possible policy has success probability $\frac1n$ by stopping at a uniformly random time.
With an additional $\alpha$-power signal, the optimal randomized policy has significantly larger success probability. %

\begin{restatable}[Randomized optimum]{theorem}{thmRandomizedOptimum}\label{thm:adversarial-rand-opt}
Let $\alpha>0$. For adversarial arrival orders with an $\alpha$-power signal, 
the optimal success probability for randomized policies is
	\(
		\opt^{\mathrm{rand}}_{n}(\alpha)
		=
		n^\alpha / ( \sum_{j=1}^n j^\alpha ).
	\)
    The policy $\pol(\max\{R,S\})$ is optimal where $R\in[n]$ is independent of~$S$ and, for $r\in[n]$, 
	\[
		\pr[R \le r]
		=
		\frac{n^\alpha}{\sum_{j=1}^n j^\alpha}
		\cdot
		\frac{\sum_{j=1}^r j^\alpha}{r^\alpha} \,.
	\]        
\end{restatable}

We can again analyze the asymptotic behavior of $\optrand_n(\alpha)$ for $n \to \infty$. 
\begin{restatable}{corollary}{coroRandomizedOptimum}\label{cor:adversarial-rand-regimes}
	For adversarial orders with an $\alpha$-power signal,
    $\opt^{\mathrm{rand}}_{n}(\alpha)=o(1)$ if $\alpha = o(n)$
    and
    \(
		\opt^{\mathrm{rand}}_{n}(\alpha) \to 1-e^{-c} 
	\) for $\alpha/n\to c\in(0,\infty)$.
\end{restatable}

\subsection{Deterministic policies}

Without any information augmentation, it is easy to see that no
deterministic policy can achieve positive success probability, even for $n=2$.
In the $\alpha$-power signal model this picture changes completely.
We show that even for small values of~$\alpha$ deterministic policies can achieve
positive success probability and, asymptotically, they even match
the performance of randomized policies for large enough values of $\alpha$.
Note that, in our model, no policy gains by stopping before the signal.

In the $\alpha$-power model, the deterministic optimum again admits a closed-form exact characterization.

\begin{restatable}[Deterministic optimum]{theorem}{thmDeterministicOptimum}\label{thm:adversial-det}
	Let $\alpha>0$. 
    For adversarial orders with an $\alpha$-power signal, policy $\pol(S)$ is optimal with success probability
	\(
		\opt^{\mathrm{det}}_{n}(\alpha)
		=
		1-\left(1-\frac{1}{n}\right)^\alpha.
	\)
\end{restatable}

Turning again to the limit $n\to \infty$, we observe that $\optdet_n(\alpha)$ behaves as $\optrand_n(\alpha)$. 
\begin{restatable}{corollary}{coroDeterministicOptimum}\label{cor:adversarial-det-regimes}
	For the adversarial orders with an $\alpha$-power signal it holds that 
		 if $\alpha=o(n)$, then $\opt^{\mathrm{det}}_{n}(\alpha)=o(1)$, and 
         if $\alpha/n\to c\in(0,\infty)$, then
		      \(
			      \opt^{\mathrm{det}}_{n}(\alpha) \to 1-e^{-c}.
		      \)
\end{restatable}

\section{Multiple uniform signals versus one $\alpha$-power signal}
\label{sec:multiple-vs-alpha}

In this section, we investigate whether a policy for adversarial arrival can exploit receiving $m$ i.i.d. uniform signals in $[i^*]$
instead of receiving a single $m$-power signal, which follows the law of the maximum of those independent signals.
Let
$S_1,\ldots,S_m$ be independently sampled from the uniform distribution over $[i^\star]$ and
$L:=\max_{j\in[m]} S_j$.
Thus, $L$ satisfies $\pr[L\le \ell \mid i^\star=i] = \left(\frac{\ell}{i}\right)^m$.
Revealing only signal $L$ is exactly the same as revealing a single $\alpha$-power signal with parameter
$\alpha=m$.

\subsection{Randomized policies: the last signal is enough}

We first show that for randomized policies under adversarial arrival, it does not matter whether all $m$ signals are known or only the latest signal is known. 
That is, the full signal history does not imply additional worst-case power beyond
the single $\alpha$-power signal with $\alpha=m$. 
For $m\in \mathbb N$, let $\optrandfull(m)$ denote the optimal success probability \emph{with} access to all~$m$ signals.

\begin{restatable}{theorem}{thmRandomizedLastSignal}\label{lem:last-signal-rand}
	For all integers $n,m \geq 1$, it holds that
	\(
		\optrand_{n}(m)
		=
		\optrandfull(m).
	\)
\end{restatable}

\subsection{Deterministic policies: the full history can help}

We move to deterministic policies. In contrast to randomized policies, we will see that 
using the full signal history can indeed improve over the optimal deterministic success probability.

Before moving to our formal results, we consider a small example. 
Let $n=4$ and $m=2$.
    With only the last signal, the optimal deterministic guarantee by \Cref{thm:adversial-det} is
	\(
		1-\left(\frac{3}{4}\right)^2=\frac{7}{16}.
	\)
    We define a deterministic policy by mapping each signal pair $(a,b)$ with $a \leq b$ 
    to a stopping time $t \geq b$:
 	\begin{alignat*}{8}	
		(1,1) &\mapsto 1, 
		& \qquad (1,2) &\mapsto 2, & \qquad 
		(1,3),(2,3),(3,3) &\mapsto 3, & \qquad 
        &\text{rest } &\mapsto 4.
    \end{alignat*}
	The resulting success probabilities on 
    the worst-case instances $\mathcal I_1,\mathcal I_2,\mathcal I_3,\mathcal I_4$ are
	$1, \frac12, \frac59, \frac12$,
	respectively. Hence, the worst-case success probability of this deterministic policy is $\frac12 > \frac{7}{16}$.

    For general $m,n\in\mathbb N$, we now characterize $\optdetfull(m)$, the optimal success probability \emph{with} access to all~$m$ signals, using an integer linear program (ILP). %
    Clearly, $\optdet_n(m) \leq \optdetfull(m)$.

In this richer setting, we define the \emph{history} as a vector $h=(c_1,\ldots,c_{\ell(h)},0,\ldots) \in (\{0\}\cup[m])^n$ with $c_{\ell(h)}>0$ and $\sum_{t=1}^{\ell(h)} c_t = m$, where $c_t \in [m]$ is
the number of signals at time $t$ and $\ell(h) \in [n]$ is the time of the last signal. 
Let $\mathcal H_m$ denote the set of all possible histories.
Setting $0! := 1$, we show that, on $\cI_i$, history $h \in \cH_m$ is observed with probability
\[
	\lambda_i(h) :=
	\begin{cases}
		 \frac{m!}{c_1! \cdot \ldots \cdot c_{\ell(h)}!}\cdot \frac{1}{i^m} & \text{if } \ell(h)\le i, \\
		 0                                                     & \text{if } \ell(h)>i.
	\end{cases}
\]

In the following ILP, observing variable $x_{h,t}=1$ means that the policy stops after history~$h$ at time~$t$. 
\begin{alignat}{4}
		\max \quad   & z    \hypertarget{ilp}{\tag{ILP}}                                                                                             \\
		\text{s.t.}\quad && \sum_{t=\ell(h)}^n x_{h,t} 
            & =1  & \quad & \forall h \in \mathcal H  \notag                           \\
		&& \sum_{h \in \mathcal H:\ell(h)\le i} \lambda_i(h)\,x_{h,i} 
            & \ge z &  & \forall i\in[n] \notag       \\
		&&  x_{h,t} 
            & \in\{0,1\}   &  & \forall h \in \mathcal H,\ \forall t\in\{\ell(h),\ldots,n\} \notag
\end{alignat}

\begin{restatable}[ILP characterization]{theorem}{thmILP}\label{thm:det-ilp}
	The optimal objective value of $(\hyperlink{ilp}{\ILP})$ is equal to $\optdetfull(m)$.
\end{restatable}

For the special case of~$m = 2$ we give a full characterization of $\optdetfull(2)$ in \Cref{app:subsec:det-full}.
In particular, we show  
	\(
		\frac{6(n-1)}{(n+1)(2n+1)}
		\le
		\optdetfull(2) 
		\le
		\frac{6n}{(n+1)(2n+1)}
	\)
which implies $\optdetfull(2) =\frac{3}{n}+O\left(\frac{1}{n^2}\right)$. 
For small values of $n$, this notably improves over $\opt_n^{\mathrm{det}}(2) = \frac{2}{n} + O(\frac{1}{n^2})$.

\section{Numerical experiments}\label{sec:experiments}

We complement the theoretical results with synthetic experiments in the random-order model.
In \Cref{app:additional-experiments} we present more in-depth experiments, also in the adversarial-order model.

\paragraph{Setup.}
For $n=1000$ and signal parameter $\alpha$, we first sample a uniformly random permutation of ranks $[n]$ and let $I$ be the arrival time of the highest rank.
In the clean model, the signal $S \in[I]$ is then sampled from
$\Pr[S=s\mid I=i] = \frac{s^\alpha-(s-1)^\alpha}{i^\alpha}$.
The empirical success probabilities of the actual policies are averaged over $1000$ trials for each specific parameterization.
Shaded regions are approximate $95\%$ confidence intervals.
The middle panel reports empirical paired gains over the classic $\frac1e$-baseline policy; standard errors for these gains are included in the supplementary files.

\begin{figure}[t]
    \centering
    \includegraphics[width=\textwidth]{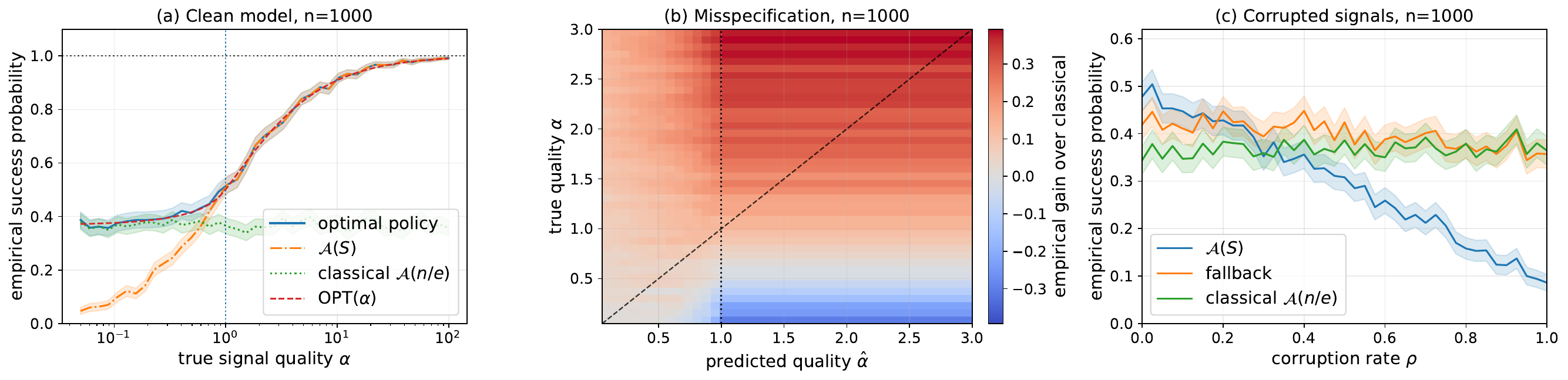}
    \caption{Empirical experiments for asynchronous predictions. Left: clean-model empirical success probabilities. Middle: empirical gain over the classic baseline under misspecified signal quality. Right: robustness to corrupted distributions.}
    \label{fig:main-experiments}
\end{figure}

\paragraph{Clean finite-sample behavior.}
In the left panel of \Cref{fig:main-experiments}, we compare the optimal policy from \Cref{thm:random-order-tight}, the signal-trusting policy $\pol(S)$, and the classic $\pol\big(\frac ne\big)$ policy. %
The function $\opt(\alpha)$ is shown for reference.
The simulation confirms that the asymptotic improvement over the classic threshold policy $\pol\big(\frac ne\big)$ is also visible at this finite scale; we present more results for smaller $n$ in \Cref{app:additional-experiments}.
It also illustrates the transition at $\alpha=1$: for late signals, $\pol(S)$ is optimal, while for early signals it is too aggressive and an additional waiting threshold is needed.

\paragraph{Misspecified signal quality.}
The middle panel evaluates the learning-augmented policy $\pol(\max\{S,k_n\})$ obtained by 
tuning the threshold $k_n$ to a predicted parameter $\hat\alpha$ as in \Cref{thm:conservative-misspecification}.
For each true $\alpha$, we generate $1000$ trials and 
evaluate all $\hat\alpha$-tuned policies on those same instances.
The heatmap plots the empirical success probability minus
the empirical success probability of the classic $\pol\big(\frac ne\big)$ policy.
The dashed diagonal is the correctly specified case, and the vertical line at $\hat\alpha=1$ marks the point after which all larger predictions induce the same policy $\pol(S)$.
The conservative region, where $\hat\alpha\le \alpha$, remains positive in accordance with \Cref{thm:conservative-misspecification}; overconfident predictions can be harmful when the true signal is early, 
but we only observe a smooth degradation.

\paragraph{Corrupted asynchronous signals.}
The right panel uses $\alpha=1$ and considers different \emph{levels of corruption}.
With probability $1-\rho$ the clean signal is observed.
With probability $\rho$, the signal is corrupted in one of three equally likely ways:
it is missed entirely,
replaced by a uniformly random false alarm in $[n]$,
or delayed to a uniformly random time after $I$ when such a time exists.
We compare the optimal policy $\pol(S)$ for $\alpha=1$ that trusts the signal, the classic $\pol\big(\frac ne\big)$ policy, and a fallback policy, which uses the optimal policy $\pol(S)$ that trusts the signal if the signal arrives before the classic threshold $\big\lceil \frac ne \big\rceil$, but otherwise reverts to the classic $\pol\big(\frac ne\big)$ policy.
The figure shows that blindly trusting corrupted signals can degrade rapidly, whereas the fallback policy gives a smoother interpolation between the learned and prediction-free regimes.

\section{Discussion and limitations}
\label{sec:discussion-limitations}

We study a deliberately weak form of side information: 
the algorithm receives no a priori prediction, but only an asynchronously delayed signal that arrives before the best item.
Our results show that even this timing-only information can be algorithmically useful.
In the random order model, it breaks the classic $\frac1e$ barrier via a simple threshold rule, 
and conservative underestimation of the signal quality remains fully robust.
In adversarial order, timing information alone is weaker, but randomization and access to signal history still provide meaningful improvements.

The main limitation is that our sharp theory assumes a clean signal: 
it always precedes the best item and follows the $\alpha$-power model.
Real predictions may be missing, delayed, or triggered by false positives.
Our experiments indicate that a fallback policy can mitigate such corruptions, 
but a full theoretical treatment of noisy asynchronous signals is left for future work.
Other natural extensions include richer secretary models with more general feasibility constraints.

\printbibliography

\appendix
\crefalias{section}{appendix}
\crefname{appendix}{Appendix}{Appendices}
\Crefname{appendix}{Appendix}{Appendices}

\section{Omitted proofs from \Cref{sec:random-order-optimal}}\label{app:random-order-optimal}
This section is dedicated to giving the formal proofs omitted from \Cref{sec:random-order-optimal}. 
We start by proving the Bellman recursion. 

\lemmaBellman*

\begin{proof}
	Stopping before the signal can never succeed because $S\le I$ almost surely.
	Likewise, stopping at a non-record cannot be optimal.
	Hence we may restrict attention to policies that only stop at records in $\{S,\ldots,n\}$.

	Fix a time $t$ and a history $\mathcal H_t$ satisfying the assumptions in the statement and let~$s$ be the realized signal time in~$\mathcal H_t$. 

    We start by bounding the probability of observing~$\hist_t$.
    The realized order~$\sigma_t$ of relative ranks up to time $t$ with~$R_t = 1$ and~$S=s$ determine $\mathcal H_t$.    
    Therefore, 
	\[
		\pr[\mathcal H_t,I=i] 
        = 
        \prob{S=s,\sigma_t, R_t =1, I = i}
        =
        \prob{S=s \mid I = i} \cdot \prob{\sigma_t \mid I=i} \cdot \prob{I=i},
    \]
    where we used that $\{S=s\}$ only depends on $\{I=i\}$ but not on~$\sigma_t$. Let \(g_s:=s^\alpha-(s-1)^\alpha. \)

    If $i>t$, the maximum has not yet appeared. Hence, the probability of observing a particular random order~$\sigma_t$ conditioned on~$\{I=i\}$ is~$\frac1{t!}$. 
    Thus, 
    \[
        \pr[\mathcal H_t,I=i] = \prob{S=s \mid I = i} \cdot \prob{\sigma_t \mid I=i} \cdot \prob{I=i}
		=
		\frac{g_s}{i^\alpha}\cdot \frac1{t!} \cdot \frac1n 
	\]
    
	If $i=t$, then the current item is the maximum. 
    Hence, the probability of observing a particular random order~$\sigma_t$ conditioned on~$\{I=t\}$ is the same as observing the prefix~$\sigma_{t-1}$. 
    Thus
     \[
        \pr[\mathcal H_t,I=t] = \prob{S=s \mid I = t} \cdot \prob{\sigma_t \mid I=t} \cdot \prob{I=t}
		=
		\frac{g_s}{t^\alpha}\cdot \frac1{(t-1)!} \cdot \frac1n 
	\]
	As the events~$\{I=i\}$ are disjoint, we can sum over~$i\in[n]$ and obtain
	\[
		\pr[\mathcal H_t]
		=
		g_s \cdot \frac1{t!} \cdot  \frac1n \cdot \left(t^{1-\alpha}+\sum_{i=t+1}^n i^{-\alpha}\right)
		=
		g_s \cdot \frac1{t!} \cdot  \frac1n \cdot  \Psi_t.
	\]
	Therefore,
	\begin{equation}
		\pr[I=t\mid \mathcal H_t] = \frac{\prob{I = t, H_t}}{\prob{H_t}} = \frac{t^{1-\alpha}}{\Psi_t}, \label{eq:bellman-succ-prob1}
	\end{equation}
	and
	\begin{equation}
		\pr[I=i\mid \mathcal H_t]=\frac{i^{-\alpha}}{\Psi_t}, \label{eq:bellman-succ-prob2}
	\end{equation}
	for all $i > t$,
	so the optimal success probability~$\Pi_t$ indeed depends only on $t$.

	If we stop at time $t$, we succeed exactly when $I=t$,
	so we win with probability $\Pi_t^{\mathrm{stop}}=\frac{t^{1-\alpha}}{\Psi_t}$.

	Let $U$ denote the next record after time $t$.
	If we continue at time $t$, then the only way to reach a future decision
	state is if $U = u > t$.
	Thus, our success probability if we continue at time $t$ is given by
	\[
		\Pi_t^{\mathrm{cont}}= \sum_{u=t+1}^n \pr[U=u \mid \mathcal H_t] \cdot \Pi_u \ .
	\]
	We next compute $\pr[U=u \mid \mathcal H_t]$.
    Observe that for any~$r \neq I$, we have $\pr[R_r \neq 1\mid r < I] = 1-\frac1r$. 
	
    If $I=u$, then $u$ is the next record after $t$ if there is no record in $\{t+1,\ldots,u-1\}$.
    Hence, we get
	\[
		\pr[U=u \mid I=u,\mathcal H_t] = \prod_{r=t+1}^{u-1} \biggl( 1 - \frac1r \biggr) = \frac{t}{u-1}
	\]
	If $I>u$, then $u$ is the next record after $t$ if there is no record in $\{t+1,\ldots,u-1\}$ and $u$ is a record itself, which happens with probability~$\frac1u$ conditioned on~$I > u$. Thus,
	\[
		\pr[U=u \mid I>u,\mathcal H_t] = \frac{1}{u} \prod_{r=t+1}^{u-1} \biggl( 1 - \frac1r \biggr) = \frac{t}{u(u-1)}
	\]
	Combining both with \eqref{eq:bellman-succ-prob1} and \eqref{eq:bellman-succ-prob2} gives
	\begin{align*}
		\pr[U=u \mid \mathcal H_t]
		 & = \pr[U=u \mid I=u,\mathcal H_t] \cdot \pr[I=u \mid \mathcal H_t] + \pr[U=u \mid I>u,\mathcal H_t] \sum_{i=u+1}^n \pr[I=i \mid \mathcal H_t] \\
		 & = \pr[U=u \mid I=u,\mathcal H_t] \cdot \frac{u^{-\alpha}}{\Psi_t} + \pr[U=u \mid I>u,\mathcal H_t] \sum_{i=u+1}^n \frac{i^{-\alpha}}{\Psi_t}       \\
		 & = \frac{t}{u(u-1)\Psi_t} \biggl(u^{1-\alpha} + \sum_{i=u+1}^n i^{-\alpha} \biggr)                                                               \\
		 & = \frac{t}{u(u-1)\Psi_t} \Psi_u \ .
	\end{align*}
	Since we already established that $\Pi_u$ only depends on $u$ for all $u \geq t$, we have
	\[
		\Pi_t^{\mathrm{cont}}= \sum_{u=t+1}^n \pr[U=u \mid \mathcal H_t] \cdot \Pi_u
		= \sum_{u=t+1}^n \frac{t}{u(u-1)\Psi_t} \Psi_u \cdot \Pi_u
		= \frac{1}{\Psi_t} \sum_{u=t+1}^n \frac{t}{u(u-1)} \Phi_u \ .
	\]

	In total, taking the maximum over both actions yields
	\[
		\Pi_t = \max\{\Pi_t^{\mathrm{stop}}, \Pi_t^{\mathrm{cont}}\}
		= \max \biggl\{ \frac{t^{1-\alpha}}{\Psi_t}, \frac{1}{\Psi_t} \sum_{u=t+1}^n \frac{t}{u(u-1)} \Phi_u \biggr\} \ .
	\]
	Multiplying by $\Psi_t$ gives \eqref{eq:bellman-random-order} and completes the proof of the lemma. 
\end{proof}

To solve the recursion just proved, we characterize the success probability of the previously defined threshold policies $\pol(\max\{S,k\})$. 
Recall that $\pol(\max\{S,k\})$ accepts the first record at or after the signal, but not before~$k$. 

\begin{restatable}[Threshold policy characterization]{lemma}{lemmaThresholdCharacterization}\label{lem:max-threshold-characterization}
	Let $k\in\{2,\ldots,n\}$.
	The success probability of the threshold policy $\pol(\max\{S,k\})$ is
	\[
		p_{n,\alpha}(k)
		=
		\frac{1}{n}\sum_{i=k}^n
		\left(
		1-\frac{1}{(i-1)i^\alpha}\sum_{r=k}^{i-1} r^\alpha
		\right).
	\]
    For $k = 1$, the success probability of $\pol(\max\{S,1\})$ is $p_{n,\alpha}(1) = p_{n,\alpha}(2) + \frac1n$.
\end{restatable}

\begin{proof}
    First, fix $k \in [n]$. We condition on $I=i$. 
    
	If $i<k$, then $\pol(\max\{S,k\})$ clearly cannot win, and these events do not contribute to $p_{n,\alpha}(k)$. 
    Hence, for $i=1$, only $\pol(\max\{S,1\})$ has a positive probability of winning; it succeeds with probability~$1$ if~$I=1$. 
    
    Now, assume $i\ge \max\{2,k\}$ and additionally condition on $S=s$.
    The policy starts accepting at time $\max\{s,k\}$.     
	It succeeds if and only if the largest item among the first $i-1$ arrivals appears before time $\max\{s,k\}$.
	Since its location is uniform in $\{1,\ldots,i-1\}$, we obtain
	\[
		\pr[\text{success}\mid I=i,S=s]=\frac{\max\{s,k\}-1}{i-1}.
	\]
	With the law of total probability, i.e., by averaging over $S$, we obtain
	\begin{align*}
		\pr[\text{success}\mid I=i]                                                   
		 & =
		\sum_{s=1}^i \frac{s^\alpha-(s-1)^\alpha}{i^\alpha}\cdot \frac{\max\{s,k\}-1}{i-1} \\
		 & =
		\frac{k-1}{(i-1)i^\alpha}\sum_{s=1}^{k-1}\bigl(s^\alpha-(s-1)^\alpha\bigr)
		+
		\frac{1}{(i-1)i^\alpha}\sum_{s=k}^i (s-1)\bigl(s^\alpha-(s-1)^\alpha\bigr).
	\end{align*}
	The first sum telescopes to $(k-1)^\alpha$.
	For the second sum,
	\[
		\sum_{s=k}^i (s-1)\bigl(s^\alpha-(s-1)^\alpha\bigr)
		=
		(i-1)i^\alpha-(k-1)^{\alpha+1}-\sum_{r=k}^{i-1} r^\alpha.
	\]
	Hence,
	\[
		\pr[\text{success}\mid I=i]
		=
		1-\frac{1}{(i-1)i^\alpha}\sum_{r=k}^{i-1} r^\alpha.
	\]
	for $i\geq 2$. 
    Averaging over the uniform choice of $I$ proves the statement if $k \geq 2$.

    For $k=1$, recall that 
    \(
        \prob{\success \mid I = 1} = 1 . 
    \)
    Hence, averaging again over~$I$ concludes the proof.
\end{proof}

We can now finally characterize the optimal policies in both regimes, $\alpha \geq 1$ and $0 < \alpha < 1$. 

\thmOptimalRandomOrder*

\begin{proof}
	We solve the Bellman recursion from \Cref{lem:random-order-bellman}.

	\smallskip
	\noindent\textbf{Case 1: $\alpha\ge 1$.}
	We claim that $\Phi_t=t^{1-\alpha}$ for every $t$, i.e.\ it is optimal to stop at the first record after the signal.
	The proof is by backward induction on $t$.
	The claim is clear for $t=n$.
	Assume it holds for all $u \geq t + 1 \geq 2$.
	Then the term representing continuing (and not stopping) in~\eqref{eq:bellman-random-order} equals
	\[
		\sum_{u=t+1}^n \frac{t}{u(u-1)}\Phi_u
		=
		t\sum_{u=t+1}^n \frac{1}{(u-1)u^\alpha}.
	\]
	By Bernoulli's inequality, we have $\frac{1}{u-1} \leq \big(1+ \frac{1}{u-1})^\alpha - 1 $ for all $u \geq 2$. Multiplying by $u^{-\alpha}$ gives
	\[
		\frac{1}{(u-1)u^\alpha}
		\le \frac{1}{(u-1)^{\alpha}}-\frac1{u^{\alpha}}.
	\]
	Inserting in the previous inequality yields
	\[
		t\sum_{u=t+1}^n \frac{1}{(u-1)u^\alpha}
		\le
		t\sum_{u=t+1}^n \left( \frac1{(u-1)^\alpha} - \frac1{u^\alpha} \right) %
		=
		t\bigl(t^{-\alpha}-n^{-\alpha}\bigr)
		< t^{1-\alpha}.
	\]
	Hence, stopping is better than continuing at every time and, thus, optimal. Therefore, $\pol(S)$ is optimal.

	It remains to evaluate the success probability of $\pol(S) = \pol(\max\{S,1\})$. 
    With \cref{lem:max-threshold-characterization}, we obtain \[
        \opt_n(\alpha) = p_{n,\alpha}(1) = \frac1n + \frac1n \sum_{i=2}^n \left( 1 - \frac1{(i-1)i^\alpha} \sum_{r=1}^{i-1} r^\alpha\right). 
    \]

    In order to determine the behavior of~$\opt_n(\alpha)$ for $n\to \infty$, we observe that 
    \[
		\sum_{r=1}^{i-1} r^\alpha = \frac{i^{\alpha+1}}{\alpha+1}+O(i^\alpha),
	\]
    for $i \ge 2$.
    Hence,
    \[
       1 - \frac1{(i-1)i^\alpha} \sum_{r=1}^{i-1} r^\alpha = 1 - \left(  \frac{i^{\alpha+1}}{\alpha+1}+O(i^\alpha)\right) = \frac\alpha{\alpha+1} + O(i^{-1}). 
    \]
    Overall, 
    \[
        \opt_n(\alpha) = \frac{\alpha}{\alpha+1}+O\!\left(\frac{\log n}{n}\right),
    \]
    which proves the limiting behavior for $\alpha \ge 1$.

	\smallskip
	\noindent\textbf{Case 2: $0<\alpha<1$.} 
    We start by defining~$k_n$. 
    To this end, let 
    \(
		G_t:=t^{1-\alpha}-t\sum_{u=t+1}^n \frac{1}{(u-1)u^\alpha}
	\)
    for $t \in [n]$ and set $k_n:=\min\{t\in[n]: G_t\ge 0\}$. We will show that $k_n$ satisfies the claims, i.e., that $\pol(\max\{S,k_n\})$ is optimal and that~$\frac{k_n}{n} \xrightarrow{n\rightarrow \infty} (1-\alpha)^{1/\alpha}$. 

    To this end, we first observe that~$G_n \geq 0$ as the sum is empty, and hence, $k_n$ is well-defined. 
    We next show that the optimal policy is a threshold rule. 
    Set
	\[
		C_n:=(k_n-1)\sum_{u=k_n}^{n}\frac{1}{(u-1)u^\alpha}.
	\]
	Consider the solution to~\eqref{eq:bellman-random-order} given by
	\[
		\widehat\Phi_t=
		\begin{cases}
			C_n,          & t<k_n,    \\
			t^{1-\alpha}, & t\ge k_n.
		\end{cases}
	\]
	We verify that this satisfies the Bellman equation:  
	Observe that
	\[
		\frac{G_{t+1}}{t+1}-\frac{G_t}{t}
		=
		(t+1)^{-\alpha}-t^{-\alpha}+\frac{1}{t(t+1)^\alpha}
		=
		\frac{(t+1)^{1-\alpha}-t^{1-\alpha}}{t}>0.
	\]
	Hence, $\frac{G_t}t$ is strictly increasing in $t$, and~$G_t \geq 0$ for $t\geq k_n$ due to this monotonicity and the definition of~$k_n$. 
    Thus, for $t \geq k_n$, 
	\[
		t^{1-\alpha}
		\ge
		t\sum_{u=t+1}^n \frac{1}{(u-1)u^\alpha}
		=
		\sum_{u=t+1}^n \frac{t}{u(u-1)}\widehat\Phi_u
	\]
	implying that stopping is optimal. 
    Therefore, for $k_n = 1$, we have just shown that $\pol(\max\{S,1\})$ is optimal.

    Now suppose that $k_n > 1$. If $t<k_n$, then
	\begin{align*}
		\sum_{u=t+1}^n \frac{t}{u(u-1)}\widehat\Phi_u
		 & =
		\sum_{u=t+1}^{k_n-1}\frac{t}{u(u-1)}C_n
		+
		\sum_{u=k_n}^{n}\frac{t}{u(u-1)}u^{1-\alpha} \\
		 & =
		t\left(\frac1t-\frac{1}{k_n-1}\right)C_n
		+
		t\sum_{u=k_n}^{n}\frac{1}{(u-1)u^\alpha}     \\
		 & =
		C_n
		=
		\widehat\Phi_t.
	\end{align*}
	Moreover, $G_{k_n-1}<0$ by definition of~$k_n$. 
    Hence, 
	\[
		C_n
		=
		(k_n-1)\sum_{u=k_n}^{n}\frac{1}{(u-1)u^\alpha}
		>
		(k_n-1)^{1-\alpha}.
	\]
	Since $t^{1-\alpha}$ is increasing for $0<\alpha<1$, we get
	\[
		\widehat\Phi_t=C_n>t^{1-\alpha}
	\]
	for all $t<k_n$.
	Thus, continuing is optimal for $t < k_n$.
	Combined with the restriction from \Cref{lem:random-order-bellman} that one only stops at records at or after the signal, policy $\pol(\max\{S,k_n\})$ behaves as desired.

	It remains to determine the asymptotic behavior of $k_n$.
	Fix $c\in(0,1]$ and let $t_n=\floor*{cn}$.
	Hence,
	\begin{equation*}
		n^{\alpha-1}G_{t_n}
        =
        \left( \frac{t_n}{n}\right)^{1-\alpha} - \frac{t_n}{n} \sum_{u=t_n+1}^n \frac{ n^{\alpha} }{(u-1)u^\alpha} 
        =
        \left( \frac{t_n}{n}\right)^{1-\alpha} - \frac{t_n}n \cdot \frac 1n \sum_{u=t_n+1}^n \left( \frac{u}{n} \right)^{-(\alpha+1)} \frac{u}{u-1} . 
    \end{equation*}    
    Thus, in the limit,
    \begin{equation*}
		n^{\alpha-1}G_{t_n} \xrightarrow{n \to \infty}
		c ^{1-\alpha}-c\int_c^1 y^{-(\alpha+1)}\,dy 
        =
        c^{1-\alpha}-c\frac{c^{-\alpha}-1}{\alpha}
        =
		\frac{c-(1-\alpha)c^{1-\alpha}}{\alpha}.
	\end{equation*}
    The expression on the right-hand side is equal to $0$ if and only if $c = (1-\alpha)^{1/\alpha}$, while $n^{\alpha-1}G_{t_n} = 0$ if and only if $G_{t_n} = 0$. 
    Moreover, since $\frac{G_t}t$ is strictly increasing, it follows that
	\[
		\frac{k_n}{n}\xrightarrow{n\to\infty}(1-\alpha)^{1/\alpha}=:c_\alpha.
	\]

	Finally, for $n$ sufficiently large, $k_n \geq 2$, and we can apply \Cref{lem:max-threshold-characterization} with~$k = k_n$ to obtain
	\begin{align*}
		\opt_n(\alpha)
		  = 
		p_{n,\alpha}(k_n) 
        =
		\frac1n\sum_{i=k_n}^n
		\left(
		1-\frac{1}{(i-1)i^\alpha}\sum_{r=k_n}^{i-1} r^\alpha
		\right)   
    \end{align*}
    Focus on the inner sum. Using $y = \frac in$ gives
    \[
    \frac{1}{(i-1)i^\alpha}\sum_{r=k_n}^{i-1} r^\alpha
    =
    \frac{n^{\alpha+1}}{(i-1)i^\alpha}\sum_{r=k_n}^{i-1} \frac{1}{n} \left(\frac{r}{n}\right)^\alpha
    \xrightarrow{n \to \infty}
    \frac{1}{y^{\alpha+1}}\int_{c_\alpha}^y z^\alpha\,dz \ .
    \]
    Now we have for the outer sum
    \[
    \frac1n\sum_{i=k_n}^n
		\left(
		1-\frac{1}{(i-1)i^\alpha}\sum_{r=k_n}^{i-1} r^\alpha
		\right)
		 \xrightarrow{n\to\infty}  
		\int_{c_\alpha}^{1}
		\left(
		1-\frac{1}{y^{\alpha+1}}\int_{c_\alpha}^{y} z^\alpha\,dz
		\right)dy .
   \]
   Evaluating the inner integral gives
   \(
        \int_{c_\alpha}^{y} z^\alpha\,dz = \frac{y^{\alpha+1} - c_\alpha^{\alpha+1} }{\alpha+1}.
   \)
   Thus the whole integrand of the outer integral becomes
   \[
        1 - \frac{1}{y^{\alpha+1}} \cdot \frac{y^{\alpha+1} - c_\alpha^{\alpha+1}}{\alpha+1}
        =
        \frac{\alpha}{\alpha+1} + \frac{c_\alpha^{\alpha+1}}{(\alpha+1)y^{\alpha+1}} \ .
   \]
   Hence the outer integral evaluates to
    \begin{align*}
		\frac{\alpha}{\alpha+1}(1-c_\alpha)
		+
		\frac{c_\alpha-c_\alpha^{\alpha+1}}{\alpha(\alpha+1)} 
      & = 
        \frac{\alpha^2 + (1-\alpha^2)c_\alpha - c_\alpha^{\alpha+1}}{\alpha(\alpha+1)} \\
      & = 
        \frac{\alpha^2 + (1+\alpha)(1-\alpha)^{1+1/\alpha} - (1-\alpha)^{1 + 1/\alpha}}{\alpha(\alpha+1)}
        =
		\frac{\alpha+(1-\alpha)^{1+1/\alpha}}{\alpha+1}.
	\end{align*}
	  This proves the claimed limit for $0<\alpha<1$.
\end{proof}

\section{Omitted proofs from \Cref{sec:random-order-robustness}} 

We start by investigating the behavior of $\polmax{S,\lceil \beta n\rceil}$ for an arbitrary parameter $\beta \in [0,1]$. 

\begin{restatable}{lemma}{lemThresholdSuccess}\label{thm:beta-threshold-success}
	Fix $\alpha>0$ and $\beta\in[0,1]$.
	Let $k_n=\ceil*{\beta n}$ and $p_{n,\alpha}(k)$ denote the success probability of the threshold policy $\pol(\max\{S,k\})$.
	Then,
	\[
		  p_{n,\alpha}(k_n)
		\xrightarrow{n \to \infty}
		f(\alpha, \beta)
		:=
		\frac{\alpha}{\alpha+1}
		+
		\frac{1-\alpha}{\alpha}\beta
		-
		\frac{\beta^{\alpha+1}}{\alpha(\alpha+1)}.
	\]
\end{restatable}

\begin{proof}
	If $\beta=0$, then $k_n=0$ for all $n$, and the claim follows from the analysis of $\pol(S)$ from the first case in the proof of \Cref{thm:random-order-tight}.

	Assume now that $\beta>0$ and consider $n$ large enough such that $k_n \geq 2$.
    Hence, \Cref{lem:max-threshold-characterization} yields
	\[
		p_{n,\alpha}(k_n)
		=
		\frac{1}{n}\sum_{i=k_n}^n
		\left(
		1-\frac{1}{(i-1)i^\alpha}\sum_{r=k_n}^{i-1} r^\alpha
		\right).
	\]
    As before, we want to analyze $p_{n,\alpha}(k_n)$ in the limit (for $n \to \infty$). 
    To this end, we approximate the sum using the Riemann integral of an appropriate function.  
	Following the same steps as in the proof of \Cref{thm:random-order-tight} with $\lceil \beta n\rceil$ as the lower limit and with $y = \frac in$
    \[   
        p_{n,\alpha}(k_n) 
        \xrightarrow{n\to\infty}
        \int_{\beta}^1 \left(1- \frac{1}{y^{\alpha+1}}\int_{\beta}^{y} z^\alpha\,dz \right) \, dy \, .
    \]
    Evaluating the inner integral again, we obtain 
    \[
        1-\frac{1}{y^{\alpha+1}}\int_\beta^y z^\alpha\,dz
		=
		\frac{\alpha}{\alpha+1}
		+
		\frac{\beta^{\alpha+1}}{(\alpha+1)y^{\alpha+1}},
    \]
    and can evaluate the whole term to obtain 
    \begin{align*}
		f(\alpha,\beta)
		 & =
		\frac{\alpha}{\alpha+1}(1-\beta)
		+
		\frac{\beta^{\alpha+1}}{\alpha+1}\int_\beta^1 y^{-\alpha-1}\,dy         \\
		 & =
		\frac{\alpha}{\alpha+1}(1-\beta)
		+
		\frac{\beta^{\alpha+1}}{\alpha+1}\cdot \frac{\beta^{-\alpha}-1}{\alpha} \\
		 & =
		\frac{\alpha}{\alpha+1}
		+
		\frac{1-\alpha}{\alpha}\beta
		-
		\frac{\beta^{\alpha+1}}{\alpha(\alpha+1)} \ .
	\end{align*}
	This completes the proof of the lemma.
\end{proof}

Next, we show how to maximize $f(\alpha,\beta)$ as a function of~$\beta \in [0,1]$.

\begin{lemma}\label{thm:success-with-best-beta}
	For every fixed $\alpha>0$, the function $f(\alpha,\cdot)$ is maximized over $[0,1]$ by
	\[
		\beta^\star(\alpha)=
		\begin{cases}
			(1-\alpha)^{1/\alpha}, & 0<\alpha<1,  \\[0.8ex]
			0,                     & \alpha\ge 1.
		\end{cases}
	\]
    Further, $f(\alpha,\beta^\star(\alpha)) = \opt(\alpha)$. 
\end{lemma}

\begin{proof}
	By \Cref{thm:beta-threshold-success},
	\(
		\odv{f_\beta(\alpha)}{\beta}
		=
		\frac{1-\alpha-\beta^\alpha}{\alpha}
	\)
	and
	\(
		\odv[ord=2]{f_\beta(\alpha)}{\beta}
		=
		-\beta^{\alpha-1}.
	\)
    
	If $0<\alpha<1$, the unique stationary point is given by $\beta^\alpha=1-\alpha$, i.e.,
	$\beta^\star(\alpha)=(1-\alpha)^{1/\alpha}$.
    Observe that the second derivative is negative there, making $\beta^\star(\alpha)$ indeed the maximizer. 
	
    If $\alpha\ge 1$, then $1-\alpha-\beta^\alpha<0$ for every $\beta>0$, so $f(\alpha,\beta)$ is strictly decreasing on $(0,1]$ and the maximum is attained at $\beta=0$.
    
	Substituting $\beta^\star(\alpha)$ into the formula from \Cref{thm:beta-threshold-success} yields exactly $\opt(\alpha)$.
\end{proof}

We now have all ingredients together to calculate the success probability if we are only given an estimate~$\hat\alpha$ for the parameter~$\alpha$. 

\thmRobustness*

\begin{proof}
    If $0<\hat\alpha< 1$, then $\beta^\star(\hat\alpha)=(1-\hat\alpha)^{1/\hat\alpha}$ by \Cref{thm:success-with-best-beta}, so the formula $g(\alpha,\hat\alpha)$ follows by substituting this value of $\beta$ into \Cref{thm:beta-threshold-success}.
    
	If $\hat\alpha\ge 1$, then $\beta^\star(\hat\alpha)=0$, so the asymptotic success probability is $\frac\alpha{\alpha+1}$ by \Cref{thm:beta-threshold-success}.

    For the second part, we distinguish $0 < \hat\alpha < 1$ and $\hat\alpha \geq 1$ and assume that $\hat \alpha \leq \alpha$.
If $\hat\alpha\ge 1$, then
\[
	g(\alpha,\hat\alpha)
	=
	\frac{\alpha}{\alpha+1}
	\ge
	\frac{\hat\alpha}{\hat\alpha+1}
	=
	\opt(\hat\alpha)
\]
since $x\mapsto \frac{x}{x+1}$ is increasing.

Now suppose $0<\hat\alpha<1$. 
We want to show that $\alpha\mapsto g(\alpha,\hat\alpha)$ is increasing, which implies $g(\alpha,\hat\alpha) \ge g(\hat\alpha,\hat\alpha)$. 
To this end, we set
\[
	\beta=\beta^\star(\hat\alpha) = (1-\hat\alpha)^{1/\hat\alpha}.
\]
Observe that 
\begin{equation*}
    g(\alpha,\hat\alpha) 
   = \frac{\alpha}{\alpha+1}
	+
	\frac{\alpha}{1-\alpha}\beta
	-
	\frac{\beta^{\alpha+1}}{\alpha(\alpha+1)} \\
   = 
    1-\beta-\frac1\alpha\left(1-\beta-\frac{1-\beta^{\alpha+1}}{\alpha+1}\right) \\
   =  
    1-\beta-\frac1\alpha\int_\beta^1(1-t^\alpha)\,dt .
\end{equation*}

For each \(t\in(0,1)\), the function $\frac{1-t^\alpha}\alpha$ is decreasing in $\alpha$ since 
\[
	\frac{1-t^\alpha}{\alpha}
	=
	\int_t^1 u^{\alpha-1}\,du
\]
and $u^{\alpha-1}$ is decreasing in $\alpha$ for $0<u<1$. 
Hence,
\[
	\alpha\mapsto \frac1a\int_\beta^1(1-t^\alpha)\,dt
\]
is decreasing, and overall, $\alpha\mapsto g(\alpha,\hat\alpha)$ is increasing. 
Therefore, with $\alpha\ge\hat\alpha$ by assumption,
\[
	g(\alpha,\hat\alpha)\ge g(\hat\alpha,\hat\alpha) \, .
\]
Finally, $\beta = \beta^\star(\hat\alpha)$ and \cref{thm:success-with-best-beta} allow us to conclude $g(\hat\alpha,\hat\alpha) = \opt(\hat\alpha)$. 
Thus, $g(\alpha,\hat\alpha)\ge \opt(\hat\alpha)$ if $\hat \alpha \leq \alpha$.
\end{proof}

\section{Omitted proofs from \Cref{sec:adversarial}}\label{app:adversarial-order}

\subsection{Randomized optimal policy}\label{app:subsec:randomized}
\thmRandomizedOptimum*

\begin{proof}
	Let
	\[
		c_{n,\alpha}:=\frac{n^\alpha}{\sum_{j=1}^n j^\alpha} \, .
	\]
    We want to show that $\optrand_n(\alpha) = c_{n,\alpha}$. 
    To this end, we first show that $\polmax{S,R}$ guarantees success probability $c_{n,\alpha}$ on every instance of length~$n$ before proving that no policy can do better. 
    
	We have 
	\[
		F(r):=\pr[R\le r]=c_{n,\alpha}\frac{\sum_{j=1}^r j^\alpha}{r^\alpha}.
	\]

	Fix an instance whose maximum arrives at time $i^\star$. Note that $i^\star$ is always a record. 
	Since $S$ follows the $\alpha$-power distribution on~$[i^\star]$, it holds that
	\[
		\pr[S=i^\star]
		=
		1-\left(\frac{i^\star-1}{i^\star}\right)^\alpha
	\]
	and consequently
	\[
		\pr[S\le i^\star-1]
		=
		\left(\frac{i^\star-1}{i^\star}\right)^\alpha \, .
	\]
	Policy $\polmax{S,R}$ succeeds in at least the following two disjoint cases:
	If $S=i^\star$ and $R\le i^\star$ or if $S\le i^\star-1$ and $R=i^\star$, $\polmax{S,R}$ stops at~$i^\star$. 
	Thus, 
	\begin{align*}
		\pr[\text{success}]
		 & \geq 
		\left(1-\left(\frac{i^\star-1}{i^\star}\right)^\alpha\right)F(i^\star)
		+
		\left(\frac{i^\star-1}{i^\star}\right)^\alpha\bigl(F(i^\star)-F(i^\star-1)\bigr) \\
		 & =
		F(i^\star)-\left(\frac{i^\star-1}{i^\star}\right)^\alpha F(i^\star-1)               \\
		 & =
		c_{n,\alpha}\left(
		\frac{\sum_{j=1}^{i^\star} j^\alpha}{{i^\star}^\alpha}
		-
		\frac{\sum_{j=1}^{i^\star-1} j^\alpha}{{i^\star}^\alpha}
		\right)
		=
		c_{n,\alpha} \, .
	\end{align*}
    Hence, there is a policy with success probability at least $c_{n,\alpha}$ which completes the proof of the lower bound on the optimal success probability.
	
	For the matching upper bound, let $\pol$ be any randomized policy and let $p_i$ be its success
	probability on the hard instance $\mathcal I_i = (1,2,\ldots,i,0,\ldots,0)$.
	With $g_s:=s^\alpha-(s-1)^\alpha$
	for $s\in[n]$, we have $\prob{S = s} = \frac{g_s}{i^\alpha}$ on $\cI_i$.
	For a signal time $s\in[n]$ and a time $t\in[n]$, let $q_{s,t}$ be the probability that $\pol$
	stops at time $t$ on instance $\mathcal I_n$ conditioned on the signal occurring at time $s$.
	If $s\le i$, then $\mathcal I_i$ and $\mathcal I_n$ are identical up to time $i$. 
    Hence, the decisions of $\pol$ are identical up to $i$ on $\cI_i$ and $\cI_n$. 
    Thus, we can rewrite $p_i$ as 
	\[
		p_i 
        = 
        \sum_{s=1}^i \prob{\text{success on }\cI_i \mid S = s} \prob{S=s} 
		=
		\frac{1}{i^\alpha}\sum_{s=1}^i g_s\,q_{s,i}.
	\]
	Therefore,
	\begin{align*}
		\sum_{i=1}^n i^\alpha p_i
		 & =
		\sum_{i=1}^n \sum_{s=1}^i g_s\,q_{s,i}
		=
		\sum_{s=1}^n g_s \sum_{i=s}^n q_{s,i} 
		  \le
		\sum_{s=1}^n g_s
		=
		n^\alpha,
	\end{align*}
	where the inequality holds because, for each fixed signal time $s$, the policy $\pol$ can stop at most once.
	This implies
	\[
		\min_{i\in[n]} p_i
		\le
		\frac{n^\alpha}{\sum_{i=1}^n i^\alpha}
		=
		c_{n,\alpha}.
	\]
	Since the worst-case guarantee of $\pol$ is at most its minimum success probability on the family
	$\{\mathcal I_i\}_{i=1}^n$, no randomized policy can beat $c_{n,\alpha}$.
    This completes the proof of the upper bound and thus of the theorem.
\end{proof}

\coroRandomizedOptimum*

\begin{proof}
	The exact formula from \Cref{thm:adversarial-rand-opt} and the bound
	\[
		\sum_{i=1}^n i^\alpha \ge \int_0^n x^\alpha\,dx = \frac{n^{\alpha+1}}{\alpha+1}
	\]
	give
	\[
		\opt^{\mathrm{rand}}_{n}(\alpha)\le \frac{\alpha+1}{n},
	\]
	which proves the asymptotic behavior for $\alpha = o(n)$. 

	For $\frac\alpha n \to c$, write $\alpha=cn+o(n)$ for some constant $c>0$.
	Then
	\[
		\sum_{i=1}^n \frac{i^\alpha}{n^\alpha}
        =
		\sum_{i=1}^n \left(\frac{i}{n}\right)^\alpha
		=
		\sum_{j=0}^{n-1}\left(1-\frac{j}{n}\right)^\alpha \ .
    \]
    For each fixed $j$, it holds that
    \[
        \left(1-\frac{j}{n}\right)^\alpha
        =
        \exp \left( \alpha \log \left( 1 - \frac{j}{n} \right) \right)
        =
        \exp \left(- (c + o(1)) j + O\left(\frac{j^2}n\right) \right) 
        \xrightarrow{n \to \infty} \exp(-jc) \ .
    \]
    Hence,
    \[
        \sum_{i=1}^n \frac{i^\alpha}{n^\alpha}
		\xrightarrow{n \to \infty}
		\sum_{j=0}^{\infty} e^{-cj}
		=
		\frac{1}{1-e^{-c}} \, .
	\]
    Thus,
	\[
		\opt^{\mathrm{rand}}_{n}(\alpha)
		=
		\frac{n^\alpha}{\sum_{i=1}^n i^\alpha}
        = 
        \frac{n^\alpha}{n^\alpha \sum_{i=1}^n i^\alpha / n^\alpha}
		\xrightarrow{n \to \infty}
		1-e^{-c}.
		\qedhere
	\]
\end{proof}

\subsection{Deterministic optimal policy}\label{app:subsec:deterministic}
\thmDeterministicOptimum*

\begin{proof}
For the lower bound, consider policy $\pol(S)$. %
If the maximum is at position $i^\star$, this policy succeeds if $S=i^\star$, which happens with probability
\[
	\pr[S= i^\star] = 1-\left(\frac{i^\star-1}{i^\star}\right)^\alpha \ge 1-\left(1 -\frac{1}{n}\right)^\alpha.
\]
Hence,
\[
	\opt^{\mathrm{det}}_{n}(\alpha)
	\geq
    \min_{i^\star \in [n]} \pr[S=i^\star]
    =
	1-\left(1 - \frac{1}{n}\right)^\alpha,
\]
which proves the lower bound on the success probability. 

    For the upper bound, we have observed that it suffices to consider policies that do not stop before the signal.
	Fix such a policy $\pol$ and let $\tau:[n]\to[n]$ be its stopping rule on instance $\mathcal I_n = (1,2,\ldots,n-1,n)$. 
    That is, $\tau(s)\ge s$ is the time at which $\pol$ stops when the signal occurs at time $s$. 
    As before, $\cI_i$ and $\cI_n$ are indistinguishable until time $i+1$. 
	On instance $\mathcal I_i = (1,2,\ldots,i,0,\ldots,0)$, the success probability is
	\[
		p_i 
        = 
        \prob{\text{success on } \cI_i}
		=
        \sum_{s=1}^i \prob{\text{success on } \cI_i \mid S=s } \prob{S=s}
        = 
		\sum_{s\le i:\,\tau(s)=i} \frac{s^\alpha-(s-1)^\alpha}{i^\alpha} \, . 
	\]
	If the worst-case guarantee is positive, i.e., $p_i>0$ for every $i\in[n]$, then for every $i$ there
	must exist at least one signal time $s \le i$ with $\tau(s)=i$.
	Since there are exactly $n$ signal times and exactly $n$ target times, this implies that every
	$i\in[n]$ has exactly one preimage under~$\tau$. 
    Hence, $\tau$ is a bijection of $[n]$.
	Since $\tau(s)\ge s$ for all $s$, 
	\[
		\sum_{s=1}^n \tau(s) \ge \sum_{s=1}^n s.
	\]
	Because $\tau$ is a bijection, equality must hold.
	Thus, $\tau(s)=s$ for every $s\in[n]$.
	Therefore, every deterministic policy with positive worst-case guarantee agrees with $\pol(S)$ on the instance family $\{\mathcal I_i\}_{i=1}^n$.
    In particular, $p_i = \prob{\pol(S) \text{ succeeds on } \cI_i}$.
    Since $\pol(S)$ succeeds on $\cI_i$ if and only if $S = i$, this implies
	\[
		p_i
		=
		\frac{i^\alpha-(i-1)^\alpha}{i^\alpha}
		=
		1-\left(\frac{i-1}{i}\right)^\alpha.
	\]
	This quantity is decreasing in $i$. 
    Thus, the worst-case guarantee is attained at $i=n$ and equals
	\[
		1-\left(1-\frac{1}{n}\right)^\alpha 
	\]
    implying that $\pol(S)$ is indeed optimal. 
\end{proof}

\coroDeterministicOptimum*

\begin{proof}
The exact formula from \Cref{thm:adversial-det} gives
	\[
		\opt^{\mathrm{det}}_{n}(\alpha)
		=
		1-\left(1-\frac{1}{n}\right)^\alpha \ .
	\]
	If $\alpha=o(n)$, then $\frac\alpha n\to 0$ as $n\to \infty$. 
    Using $\log\left(1-\frac{1}{n}\right) =	-\frac{1}{n}+O\left(\frac{1}{n^2}\right)$,
	we get
	\[
		\alpha\log\left(1-\frac{1}{n}\right)
		=
		-\frac{\alpha}{n}+O\left(\frac{\alpha}{n^2}\right)
		=
		o(1) \ ,
	\]
	and therefore
	\[
		\opt^{\mathrm{det}}_{n}(\alpha)
		=
		1-\exp\!\left(\alpha\log\left(1-\frac{1}{n}\right)\right)
		=
		o(1) \ .
	\]

	If instead $\frac\alpha n\to c\in(0,\infty)$, then
	\[
		\alpha\log\left(1-\frac{1}{n}\right)
		=
		\frac{\alpha}{n}\cdot n\log\left(1-\frac{1}{n}\right)
		\xrightarrow{n \to \infty}
		-c \ ,
	\]
	since $n\log\big(1-\frac1n\big) \xrightarrow{n\to\infty} -1$. 
    Hence,
	\[
		\opt^{\mathrm{det}}_{n}(\alpha)
		=
		1-\exp\!\left(\alpha\log\left(1-\frac{1}{n}\right)\right)
		\xrightarrow{n \to \infty}
		1-e^{-c} \ .
		\qedhere
	\]
\end{proof}

\section{Omitted proofs from \Cref{sec:multiple-vs-alpha}}

\subsection{Randomized optimal policy}

\thmRandomizedLastSignal*

\begin{proof}
The inequality $\opt^{\mathrm{rand}}_n(m)\le \optrandfull(m)$ is immediate, since the full signal
history contains at least as much information as the last signal $L$.

For the reverse inequality, fix a randomized policy $\pol$ for the full-history
model, and consider the hard instances $\mathcal I_i = (1,2,\dots,i,0,\dots,0)$ for all $i\in[n]$.
Let $p_i$ denote the success probability of $\pol$ on $\mathcal I_i$. It is enough to show that
\[
\min_{i\in[n]} p_i \le \frac{n^m}{\sum_{j=1}^n j^m} = \optrand_n(m).
\]

Consider first a labeled signal tuple $s=(s_1,\dots,s_m)\in[n]^m$ and let
\(
L(s) := \max_{j\in[m]} s_j.
\)
As before, for $t\in\{L(s),\dots,n\}$, let $q_{s,t}$ denote the probability that
$\pol$ stops at time $t$ on $\mathcal I_n$, conditioned on the event
\(
(S_1,\dots,S_m)=s.
\)
For every fixed $s$, we have
\(
\sum_{t=L(s)}^n q_{s,t}\le 1,
\)
because $\pol$ stops with probability at most~$1$. 

Now fix $i\in[n]$. On $\mathcal I_i$, the tuple $(S_1,\dots,S_m)$ is uniformly distributed on $[i]^m$.
Moreover, for $s\in[i]^m$, the decisions of $\pol$ on $\mathcal I_i$ and
on $\mathcal I_n$ are identical up to time $i$ when conditioned on
$(S_1,\dots,S_m)=s$ since these two instances are indistinguishable and 
the same signals have appeared by time~$t \in [i]$. 
Hence,
\[
p_i=\frac{1}{i^m}\sum_{s\in[i]^m} q_{s,i}.
\]
Multiplying by $i^m$ and summing over all $i \in [n]$ gives
\[
\sum_{i=1}^n i^m p_i
=
\sum_{i=1}^n \sum_{s\in[i]^m} q_{s,i}
=
\sum_{s\in[n]^m} \sum_{i=L(s)}^n q_{s,i}
\le
\sum_{s\in[n]^m} 1
=
n^m.
\]
Therefore,
\[
\min_{i\in[n]} p_i
\le
\frac{\sum_{i=1}^n i^m p_i}{\sum_{i=1}^n i^m}
\le
\frac{n^m}{\sum_{i=1}^n i^m}.
\]
Since the worst-case guarantee of $\mathcal A$ is at most $\min_{i\in[n]} p_i$ and $\pol$ was chosen arbitrarily, we obtain
\[
\optrandfull(m)\le \frac{n^m}{\sum_{i=1}^n i^m} = \opt^{\mathrm{rand}}_n(m).
\]
by \Cref{thm:adversarial-rand-opt}, which completes the proof.
\end{proof}

\subsection{Deterministic optimal policy}\label{app:subsec:det-full}

Recall that for a history $h \in \cH_m$, we defined
\[
	\lambda_i(h) =
	\begin{cases}
		 \frac{m!}{c_1! \cdot \ldots \cdot c_{\ell(h)}!}\cdot \frac{1}{i^m}, & \text{if } \ell(h)\le i, \\
		 0,                                                     & \text{if } \ell(h)>i.
	\end{cases}
\]

\begin{restatable}{lemma}{lemHistoryProb}\label{lem:history-prob}
    Each history $h \in \mathcal H$ is observed on $\mathcal I_i$ with probability $\lambda_i(h)$.
\end{restatable}

\begin{proof}
	Fix an instance $\mathcal I_i$. In the full-history model, the $m$ signals are drawn independently and
	uniformly from $[i]$. It is convenient to temporarily consider the signals as labeled, so a signal
	realization is a tuple $s=(s_1,\ldots,s_m)\in [i]^m$.
	Each such labeled tuple has probability $i^{-m}$.

	Given a labeled tuple $s$, the observed history records only the number of signals at each time.
	That is, if $\ell(s):=\max_{j\in[m]} s_j$, then the induced history is $h(s)=(c_1(s),\ldots,c_{\ell(s)}(s))$,
    where $c_t(s)=|\{j\in[m]:s_j=t\}|$.
	Thus the policy observes the counts of signals at each time, but not the labels of the signals.
	Now fix a history $h=(c_1,\ldots,c_{\ell(h)})\in \mathcal H$.
	If $\ell(h)>i$, then no signal can occur at time $\ell(h)$ on the instance $\mathcal I_i$ since all signals
	are supported on $[i]$. Therefore $h$ is impossible, and 
	\[
		\pr[h\text{ is observed on } \mathcal I_i]=0=\lambda_i(h).
	\]

	Assume instead that $\ell(h)\le i$. The history $h$ is observed exactly when, for every
	$t\in[\ell(h)]$, precisely $c_t$ of the labeled signals are equal to $t$. Since $\sum_{t=1}^{\ell(h)} c_t=m$,
	this also accounts for all $m$ signals, so no signal occurs after time $\ell(h)$.

	It remains to count how many labeled tuples $s\in[i]^m$ induce this same unlabeled history.
	We first choose the $c_1$ labeled signals that occur at time $1$, then the $c_2$ labeled signals
	among the remaining ones that occur at time $2$, and so on. This gives
	\[
		\binom{m}{c_1}
		\binom{m-c_1}{c_2}
		\binom{m-c_1-c_2}{c_3}
		\cdots
		\binom{c_{\ell(h)}}{c_{\ell(h)}} .
	\]
	The factorials telescope, so the number of labeled tuples producing $h$ is
	\[
		\frac{m!}{c_1!\cdots c_{\ell(h)}!}.
	\]
	Here the convention $0!=1$ handles times at which no signal occurs.

	Since each labeled tuple has probability $i^{-m}$, we obtain
	\[
		\pr[h\text{ is observed on } \mathcal I_i]
		=
		\frac{m!}{c_1!\cdots c_{\ell(h)}!}\cdot \frac{1}{i^m}.
	\]
	This is exactly $\lambda_i(h)$ by definition. Hence every history $h\in\mathcal H$ is observed
	on $\mathcal I_i$ with probability $\lambda_i(h)$.
\end{proof}

Recall our ILP 
\begin{alignat}{3}
\max \quad       & z    \tag{ILP}                                                                                             \\
\text{s.t.}\quad & \sum_{t=\ell(h)}^n x_{h,t}=1  & \quad & \forall h \in \mathcal H  \notag                           \\
                 & \sum_{h \in \mathcal H:\ell(h)\le i} \lambda_i(h)\,x_{h,i} \ge z &  & \forall i\in[n] \notag       \\
                 & x_{h,t}\in\{0,1\}   &  & \forall h \in \mathcal H,\ \forall t\in\{\ell(h),\ldots,n\} \notag
\end{alignat}

\thmILP*

\begin{proof}
	First consider any deterministic policy $\pol$.
	As observed before, we may assume that $\pol$ stops only after all $m$ signals
	have appeared.
	Consider its behavior on $\mathcal I_n$.
	For every history $h$, let $\tau(h)\in\{\ell(h),\ldots,n\}$ be the stopping time of $\pol$ after
	observing $h$, and set $x_{h,t}=\mathds{1}[\tau(h)=t]$.
	Then the first set of constraints is satisfied.

	Now fix $i\in[n]$.
	Up to time $i$, the two instances $\mathcal I_i$ and $\mathcal I_n$ are identical, so after any
	history $h$ with $\ell(h)\le i$, the policy makes the same decision on $\mathcal I_i$ as on
	$\mathcal I_n$.
	Hence, by \Cref{lem:history-prob} the success probability of $\pol$ on $\mathcal I_i$ is exactly
	\[
		\sum_{h:\ell(h)\le i} \lambda_i(h)\,x_{h,i}.
	\]
	Therefore every deterministic policy induces a feasible solution to $(\ILP)$ with
	\[
		z\le \min_{i\in[n]} \pr[\pol\text{ wins on }\mathcal I_i].
	\]
	Since the worst-case guarantee of $\pol$ is at most its minimum success probability on the hard
	family $\{\mathcal I_i\}_{i=1}^n$, we obtain that
	$\opt^{\mathrm{det}}_{n,m}$ is at most the optimal objective value of $(\ILP)$.

	Conversely, let $x$ be any feasible solution to $(\ILP)$ with value $z$.
	Define a deterministic policy that ignores the item values, waits until all $m$ signals have
	appeared, reads the resulting history $h$, and then stops at the unique time $t$ with $x_{h,t}=1$.
	For any adversarial instance whose maximum is at position $i$, the distribution of the signal
	history is the same as on $\mathcal I_i$ because the signals depend only on the maximum position.
	Thus, the success probability of this policy is exactly
	\[
		\sum_{h:\,\ell(h)\le i} \lambda_i(h) \cdot x_{h,i}
	\]
	when the maximum is at position $i$, and this is at least $z$ for every $i\in[n]$ by the second set
	of constraints.
	Hence the policy has worst-case guarantee at least $z$.
	This shows that the ILP optimum is at most $\opt^{\mathrm{det}}_{n,m}$.
\end{proof}

Using the characterization just proved, we now turn to the special case of two signals. 
\begin{restatable}{theorem}{thmTwoSignals}\label{thm:m2-det}
	For $m=2$, we have
	\[
		\optdetfull(2)
		=
		\max\Bigl\{z\in[0,1]: \sum_{t=1}^\ell \lceil z t^2\rceil \le \ell^2 \text{ for all } \ell\in[n]\Bigr\}.
	\]
\end{restatable}

\begin{proof}
For $m=2$, write a history as a sorted pair $(a,b)$, $1\le a\le b\le n$, 
and let
\[
w(a,b)=
\begin{cases}
2,& a<b,\\
1,& a=b.
\end{cases}
\]
The weight $w(a,b)$ is the number of \emph{labeled} signal pairs inducing the sorted history $(a,b)$.
We may assume that a deterministic policy waits until both signals have appeared. 

In order to show 
\[
    \optdetfull(2)
  \leq
	\max\Bigl\{z\in[0,1]: \sum_{t=1}^\ell \lceil z t^2\rceil \le \ell^2 \text{ for all } \ell\in[n]\Bigr\} ,
\]
fix a deterministic policy and consider its behavior on the hard instance $\mathcal I_n$. For every history $(a,b)$, let $\tau(a,b)\ge b$ be the stopping time chosen by the policy on $\mathcal I_n$. Define
\[
C_t := \sum_{\substack{1\le a\le b\le t\\ \tau(a,b)=t}} w(a,b).
\]
On the hard instance $\mathcal I_t$, the history $(a,b)$ occurs with probability $w(a,b)/t^2$, and the policy wins exactly when $\tau(a,b)=t$. Hence its success probability on $\mathcal I_t$ is $C_t/t^2$.

Suppose the policy has worst-case guarantee at least $z$. Then
$C_t \ge zt^2$,
and since $C_t$ is an integer, $C_t \ge \lceil zt^2\rceil$.
Moreover, for every $\ell$, all histories relevant for $\cI_\ell$ must have second signal at most $\ell$. The total weight of such histories is
\[
\sum_{b=1}^\ell \sum_{a=1}^b w(a,b)
=
\sum_{b=1}^\ell (2b-1)
=
\ell^2.
\]
Therefore
\[
\sum_{t=1}^\ell \lceil zt^2\rceil
\le
\sum_{t=1}^\ell C_t
\le
\ell^2
\]
for all $\ell\in[n]$. This proves the upper bound.

For the other direction
\[
    \optdetfull(2)
  \geq
	\max\Bigl\{z\in[0,1]: \sum_{t=1}^\ell \lceil z t^2\rceil \le \ell^2 \text{ for all } \ell\in[n]\Bigr\} ,
\]
suppose that $z\in[0,1]$ satisfies
\[
\sum_{t=1}^\ell \lceil zt^2\rceil \le \ell^2
\]
for all $\ell \in [n]$.
Set
$d_t := \lceil zt^2\rceil$ for all $t<n$
and 
$d_n := n^2-\sum_{t=1}^{n-1}d_t$.
The constraint for $\ell=n$ implies $d_n\ge \lceil zn^2\rceil$.

We now construct a deterministic policy by greedily assigning histories to stopping times so that the total weight of histories assigned to time $t$ is $d_t$. 
Suppose that the assignments for times $1,\ldots,t-1$ have already been made, and let
\[
\mathcal A_t := \{(a,b):1\le a\le b\le t \text{ and } (a,b) \text{ is not yet assigned}\}
\]
be the set of currently available histories. Every history in $\mathcal A_t$ can be assigned to time $t$, since its second signal is at most $t$. Moreover, the total available weight is
\[
W_t:=\sum_{(a,b)\in \mathcal A_t} w(a,b)
= t^2-\sum_{s=1}^{t-1} d_s
\ge d_t,
\]
where the inequality follows from the prefix constraint.

It remains to justify that one can select available histories of \emph{exactly} weight $d_t$. 
The available histories have weights only $1$ and $2$, and the diagonal history $(t,t)$ is newly available at step $t$ and was not available before, so it is still unassigned and has weight $1$. 
Thus $\mathcal A_t$ contains at least one weight-$1$ history. Let $p\ge 1$ be the number of weight-$1$ histories in $\mathcal A_t$, and let $q$ be the number of weight-$2$ histories. Then $W_t=p+2q$. We claim that for every integer $r\in\{0,\ldots,W_t\}$ there is a subset $B\subseteq \mathcal A_t$ with
\[
\sum_{(a,b)\in B} w(a,b)=r.
\]
Indeed, if $r\le 2q$, write $r=2u+\varepsilon$ with $\varepsilon\in\{0,1\}$. If $\varepsilon=0$, choose $u\le q$ weight-$2$ histories. If $\varepsilon=1$, then $r\le 2q$ implies $u\le q-1$, so choose $u$ weight-$2$ histories and one weight-$1$ history. If instead $r>2q$, choose all $q$ weight-$2$ histories and $r-2q$ weight-$1$ histories; this is possible because $r\le p+2q$ implies $r-2q\le p$.

Applying the claim with $r=d_t$ gives a subset of available histories of total weight exactly $d_t$, which we assign to stopping time $t$. We do this for every $t<n$. At time $n$, assign all remaining histories to time $n$. The total remaining weight is
\[
n^2-\sum_{t=1}^{n-1} d_t=d_n,
\]
so the weight assigned to time $n$ is exactly $d_n$.
Let $\tau(a,b)$ denote the
stopping time assigned to history $(a,b)$. By construction, for every
$t\in[n]$, the total weight of histories assigned to stopping time $t$ is
exactly $d_t$, that is,
\[
\sum_{\substack{1\le a\le b\le t\\ \tau(a,b)=t}} w(a,b)=d_t.
\]

Now consider any instance whose maximum arrives at time $t$. The two signals
are independent and uniform on $[t]$. Hence a sorted history $(a,b)$ with
$a\le b\le t$ occurs with probability $w(a,b)/t^2$. The policy wins exactly
on those histories that are assigned to stopping time $t$ because then it
stops exactly when the maximum arrives. Therefore the success probability on
such an instance is
\[
\sum_{\substack{1\le a\le b\le t\\ \tau(a,b)=t}} \frac{w(a,b)}{t^2}
=
\frac{d_t}{t^2}
\ge
z \ ,
\]
where the final inequality follows from 
$d_t=\lceil zt^2\rceil\ge zt^2$ for $t<n$, and \[
d_n
=
n^2-\sum_{t=1}^{n-1}\lceil zt^2\rceil
\ge
\lceil zn^2\rceil
\ge
zn^2 \ .
\]
This defines a deterministic policy with worst-case success probability at
least $z$. Hence $z$ is achievable by a deterministic policy, so $\optdetfull(2)$ must be at least $z$.
This concludes the second part of the proof.
\end{proof}

\begin{restatable}{corollary}{coroTwoSignals}\label{cor:m2-asymptotic}
	For every integer $n\ge 4$, we have
	\(
		\frac{6(n-1)}{(n+1)(2n+1)}
		\le
		\optdetfull(2)
		\le
		\frac{6n}{(n+1)(2n+1)}.
	\)
	In particular, it holds that
	\(
		\optdetfull(2)=\frac{3}{n}+O\left(\frac{1}{n^2}\right).
	\)
\end{restatable}

\begin{proof}
	Let $z^\star:=\optdetfull(2)$.
	Then, $z^\star$ is feasible for the characterization given in \Cref{thm:m2-det}. Hence, using the constraint for $\ell=n$, we get
	\[
		z^\star \sum_{t=1}^n t^2
		\le
		\sum_{t=1}^n \ceil{z^\star t^2}
		\le
		n^2.
	\]
	Since $\sum_{t=1}^n t^2=\frac{n(n+1)(2n+1)}6$, this gives
	\[
		\optdetfull(2)=z^\star
		\le
		\frac{n^2}{\sum_{t=1}^n t^2}
		=
		\frac{6n}{(n+1)(2n+1)}.
	\]

	For the lower bound, set
	\[
		\underline z
		:=
		\frac{6(n-1)}{(n+1)(2n+1)}.
	\]
	We show that $\underline z$ is feasible for the characterization given in \cref{thm:m2-det}.
    For this, set
	\[
		g(\ell):=\frac{6(\ell-1)}{(\ell+1)(2\ell+1)}.
	\]
	A direct computation gives, for every $\ell\ge 3$,
	\[
		g(\ell+1)-g(\ell)
		=
		\frac{12(-\ell^2+\ell+3)}{(\ell+1)(\ell+2)(2\ell+1)(2\ell+3)}
		< 0.
	\]
	Thus $g$ is decreasing on $\{3,4,\ldots\}$.  
    Moreover, $g(2)=g(4)=\frac25$, so for every $n\ge 4$ and every $\ell\in\{2,\ldots,n\}$ we have
	\[
		\underline z=g(n)\le g(\ell).
	\]
	Now fix $\ell\in[n]$.  If $\ell=1$, then $0<\underline z\le 1$, and hence
	\[
		\sum_{t=1}^1 \ceil{\underline z t^2}=\ceil{\underline z}=1.
	\]
	If $\ell\ge 2$, then $\ceil{x}\le x+1$ for every real $x$, and therefore
	\[
		\sum_{t=1}^\ell \ceil{\underline z t^2}
		\le
		\underline z\sum_{t=1}^\ell t^2+\ell
		=
		\underline z\frac{\ell(\ell+1)(2\ell+1)}{6}+\ell
		\le
		g(\ell)\frac{\ell(\ell+1)(2\ell+1)}{6}+\ell
		=
		\ell^2.
	\]
	Thus $\underline z$ satisfies all feasibility inequalities in \Cref{thm:m2-det}, and so
	\[
		\opt^{\mathrm{det}}_{n,2}\ge \underline z
		=
		\frac{6(n-1)}{(n+1)(2n+1)}.
	\]

	The two bounds differ by
	\[
		\frac{6n}{(n+1)(2n+1)}-\frac{6(n-1)}{(n+1)(2n+1)}
		=
		O\left(\frac{1}{n^2}\right),
	\]
	and the upper bound is $\frac3n+O\big(\frac1{n^2}\big)$. Hence,
	\[
		\opt^{\mathrm{det}}_{n,2}
		=
		\frac{3}{n}+O\left(\frac{1}{n^2}\right).\qedhere
	\]
\end{proof}

\section{Additional numerical experiments}
\label{app:additional-experiments}

This section contains additional experiments supporting the numerical results in \Cref{sec:experiments} and the adversarial-order theory in \Cref{sec:adversarial}.  We separate the experiments into random-order experiments, which extend the main empirical evaluation, and adversarial-order experiments, which illustrate the hard instances and scaling regimes.

\subsection{Additional random-order experiments}
\label{app:random-order-experiments}

The experiments in this subsection use the same setup as in \Cref{sec:experiments}.  In each trial we sample a full random permutation of ranks $[n]$, draw the signal $S$ conditional on the realized arrival time $I$ of the maximum, run the corresponding stopping rule on the realized order, and record whether the selected item is the maximum. The plotted curves are empirical success probabilities of the actual stopping rules over $1000$ trials. Confidence bands are $95\%$ confidence intervals.

\paragraph{Model scaling with $n$.}
The first experiment repeats the model evaluation for several values of $n$.  For each $\alpha$, we use the finite-$n$ optimal threshold from \Cref{thm:random-order-tight} and estimate its success probability. \Cref{fig:appendix-clean-empirical-overview} compares these empirical curves with the asymptotic optimum $\opt(\alpha)$ and the traditional $\frac1e$ benchmark. The curves approach the asymptotic formula as $n$ grows, while already showing the improvement over the classic $\pol\big(\frac ne\big)$ policy at moderate problem sizes.

\begin{figure}[tp]
    \centering
    \includegraphics[width=.68\textwidth]{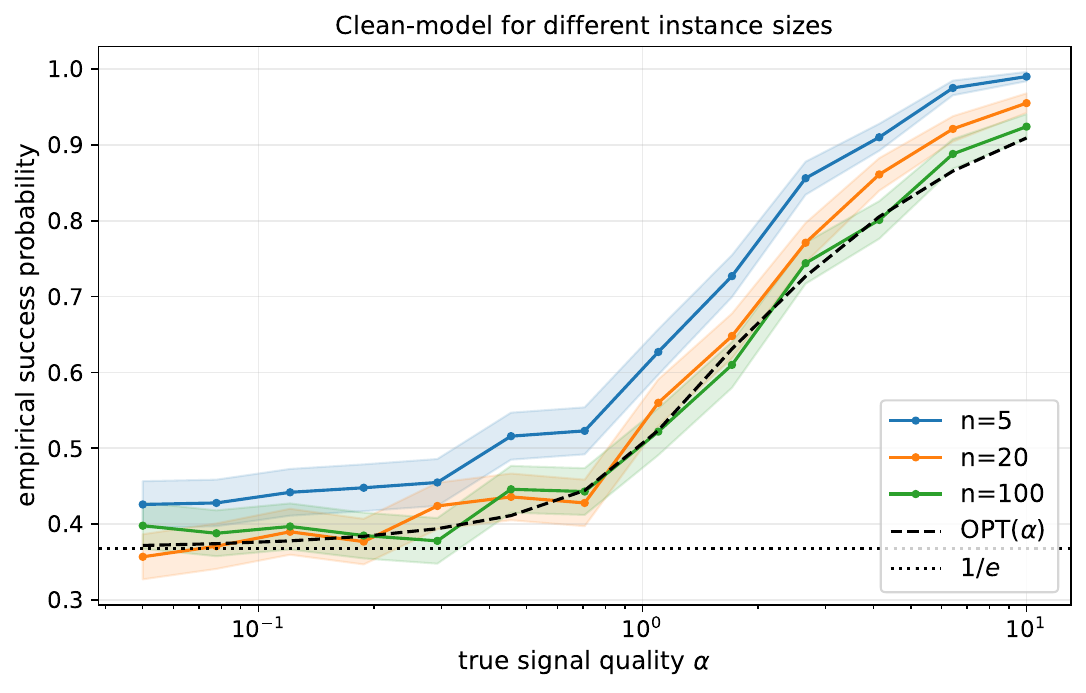}
    \caption{Random-order experiments for several problem sizes. Each curve shows the empirical success probability of the finite-$n$ optimal asynchronous threshold policy from \Cref{thm:random-order-tight}. The dashed curve is the asymptotic optimum $\opt(\alpha)$.}
    \label{fig:appendix-clean-empirical-overview}
\end{figure}

\paragraph{Misspecification slices and induced thresholds.}
The middle panel of \Cref{fig:main-experiments} summarizes robustness to misspecification as a heatmap.  \Cref{fig:appendix-misspecification-slices} gives a complementary one-dimensional view.  We fix several true values of $\alpha$ and vary the estimated parameter $\hat\alpha$. 
We consider the policy $\pol(\max\{S,k_n\})$ from \Cref{thm:conservative-misspecification}
that is tuned to the threshold
$k_n(\hat\alpha)=\max\{1,\lceil \beta^\star(\hat\alpha)n\rceil\}$.
For each true $\alpha$, the same sampled instances are used for all values of $\hat\alpha$, giving a paired comparison across thresholds. 
The lower panel plots the induced threshold fraction $\frac{k_n(\hat\alpha)}n$.  This makes explicit why the success curves become flat once $\hat\alpha\ge 1$: all such predictions induce the same signal-trusting policy $\pol(S)$.  The qualitative behavior is consistent with \Cref{thm:conservative-misspecification}: conservative tuning is robust, while sufficiently overconfident tuning can be harmful when the true signal arrives early.

\begin{figure}[tp]
    \centering
    \includegraphics[width=.68\textwidth]{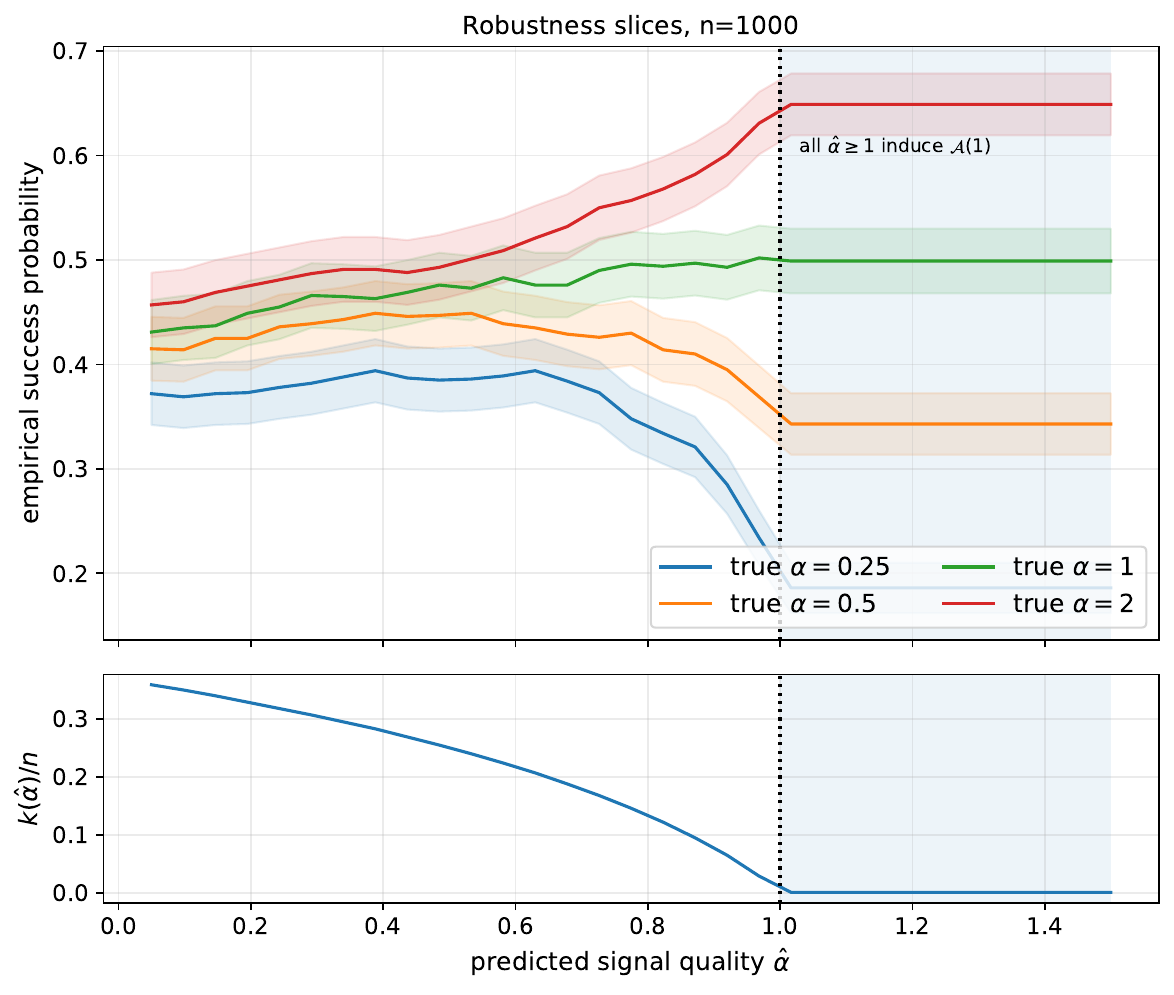}
    \caption{Misspecification slices in the random-order model. The top panel shows empirical success probabilities as a function of the predicted parameter $\hat\alpha$ for several true values of $\alpha$. The bottom panel shows the induced threshold fraction $k_n(\hat\alpha)/n$, explaining the plateau for $\hat\alpha\ge1$.}
    \label{fig:appendix-misspecification-slices}
\end{figure}

\paragraph{Separating the noisy-signal failure modes.}
The right panel of \Cref{fig:main-experiments} combines three corruption types into one mixed corruption rate.  \Cref{fig:appendix-noisy-signal-by-type} separates these effects.  In each panel, with probability $1-\rho$ the clean signal is observed, while with probability $\rho$ the signal is corrupted in one specified way: it is missed entirely, replaced by a uniformly random false alarm in $[n]$, or delayed to a uniformly random time after $I$ when such a time exists.  We compare the optimal policy from \Cref{thm:random-order-tight}, the classic threshold $\pol\big(\frac ne\big)$ policy, and the fallback policy from \Cref{sec:experiments}. 
Missed and late signals mostly harm the pure asynchronous policy by preventing it from acting in time, whereas false alarms can induce premature stopping. 
The fallback rule mitigates both effects by reverting to the classic threshold when the signal is absent or arrives too late.

\begin{figure}[tp]
    \centering
    \includegraphics[width=\textwidth]{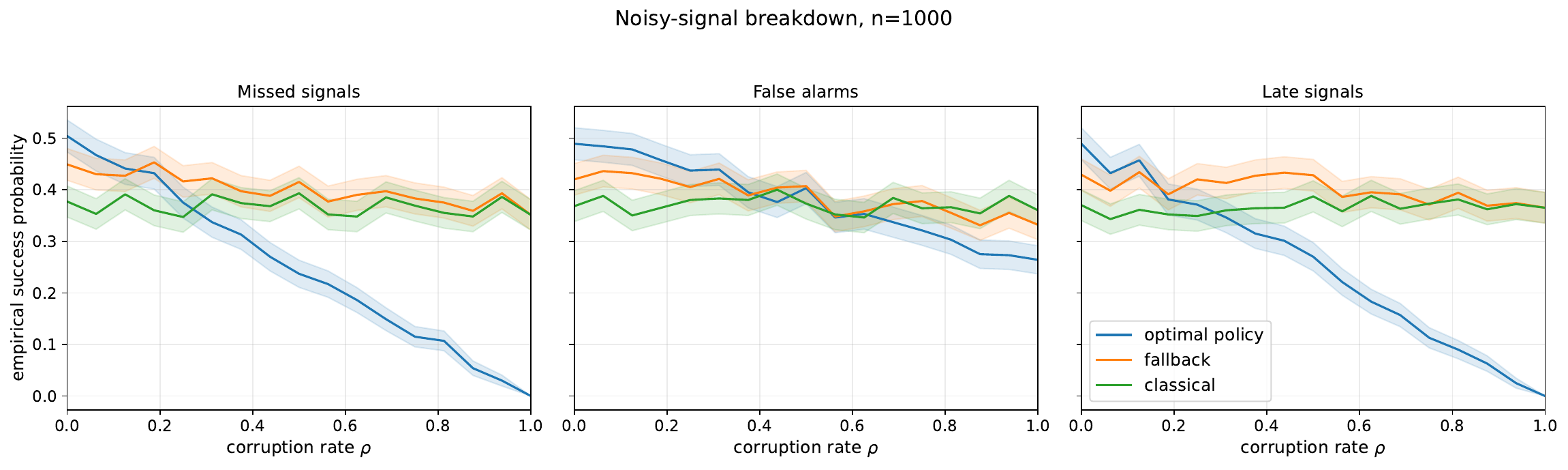}
    \caption{Noisy-signal breakdown in the random-order model. The three panels isolate missed signals, false alarms, and late signals. The fallback policy gives a smoother degradation than the pure asynchronous policy across the different corruption mechanisms.}
    \label{fig:appendix-noisy-signal-by-type}
\end{figure}

\subsection{Additional adversarial-order experiments}
\label{app:adversarial-order-experiments}

We now supplement the adversarial-order results from \Cref{sec:adversarial}.  These experiments use the hard family
$\mathcal I_i=(1,2,\ldots,i,0,\ldots,0)$
where the maximum is at position $i$ and no item after time $i$ is a record. 
Unlike the random-order experiments above, there is no permutation randomness: the adversary fixes the instance $\mathcal I_i$. 
In the single-signal experiments, randomness comes only from the signal time $S$ and, for randomized policies, from the internal random threshold of the algorithm. 
The full-history experiment for $m=2$ is computed exactly from the characterization in \Cref{thm:m2-det}.

\paragraph{Position profile of the hard family.}
The first adversarial experiment fixes $n$ and $\alpha$ and varies the adversarial maximum position $i$. 
We compare the deterministic signal-stopping rule from \Cref{thm:adversial-det} with the randomized policy from \Cref{thm:adversarial-rand-opt}. 
The deterministic policy succeeds on $\mathcal I_i$ only when the signal reaches the maximum position, so its success probability depends strongly on $i$. 
The randomized policy approximately equalizes the hard-family instances, as shown in \Cref{thm:adversarial-rand-opt}.

\begin{figure}[tp]
    \centering
    \includegraphics[width=.68\textwidth]{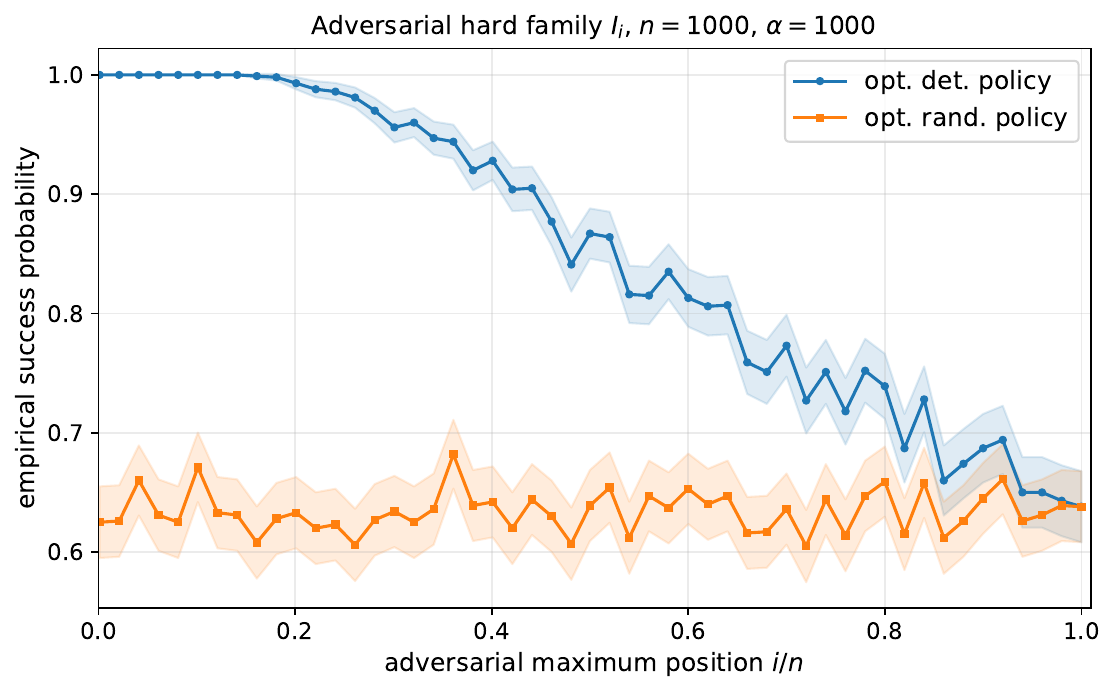}
    \caption{Adversarial position profile on the hard family $\mathcal I_i$.  The deterministic signal-stopping rule is sensitive to the adversarial maximum position, while the randomized minimax policy from \Cref{thm:adversarial-rand-opt} approximately equalizes success across positions.}
    \label{fig:appendix-adversarial-position-profile}
\end{figure}

\paragraph{Scaling regime $\alpha=cn$.}
The second adversarial experiment studies the regime $\alpha=cn$.  For each constant $c > 0$, we estimate the empirical worst-case success probability over the sampled hard-family positions and compare it with the exact deterministic and randomized values from \Cref{thm:adversial-det,thm:adversarial-rand-opt}. The limiting behavior is governed by \Cref{cor:adversarial-det-regimes,cor:adversarial-rand-regimes}: 
when $\alpha$ is proportional to $n$, both values approach the constant curve $1-e^{-c}$.  \Cref{fig:appendix-adversarial-scaling} confirms that the empirical worst-case values track the exact finite-$n$ curves and approach this limit.

\begin{figure}[tp]
    \centering
    \includegraphics[width=.68\textwidth]{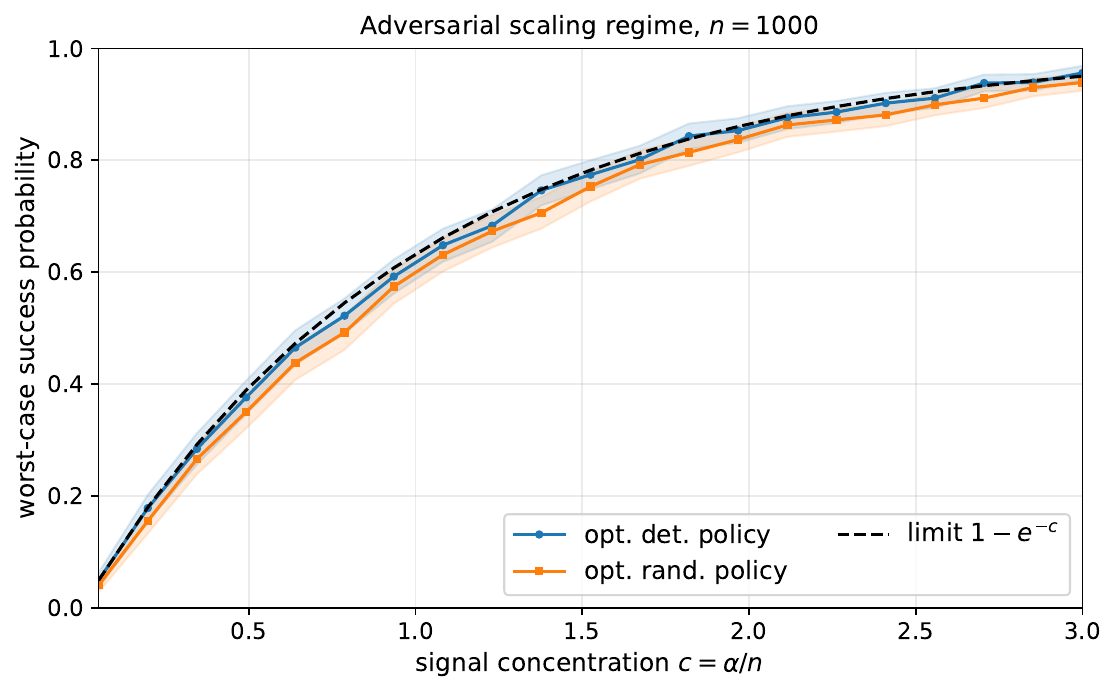}
    \caption{Adversarial scaling regime $\alpha=cn$.  Empirical worst-case success probabilities match the exact finite-$n$ deterministic and randomized values and approach the limit $1-e^{-c}$.}
    \label{fig:appendix-adversarial-scaling}
\end{figure}
\begin{figure}[tp]
    \centering
    \includegraphics[width=.68\textwidth]{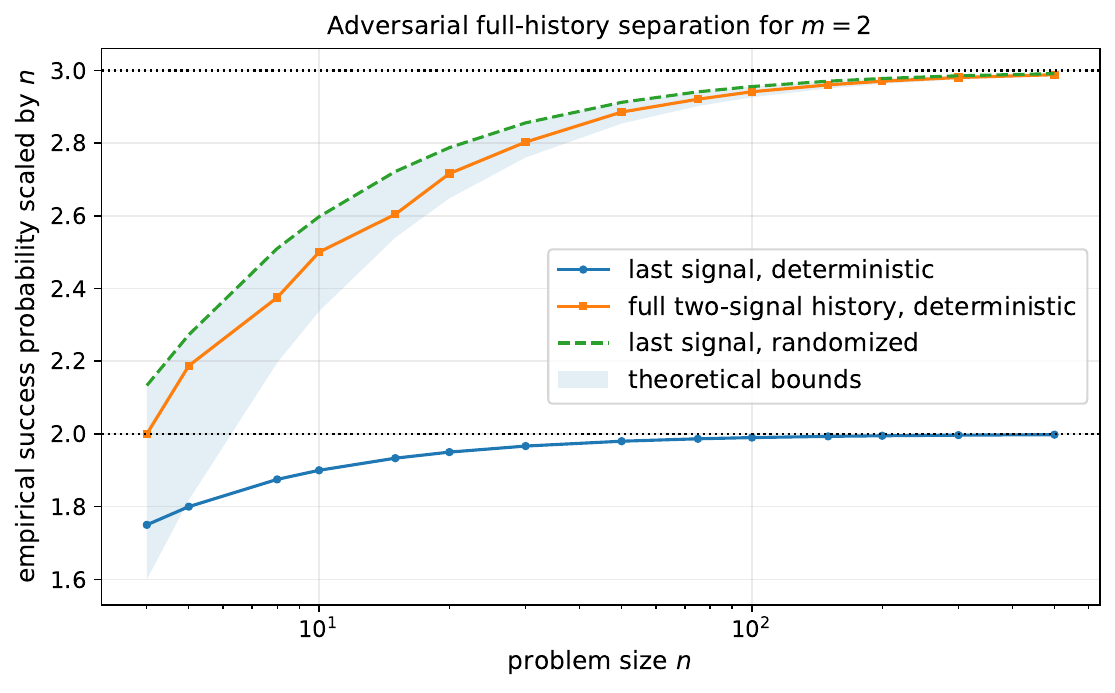}
    \caption{Full-history separation for $m=2$ in the adversarial-order model.  The full-history deterministic value from \Cref{thm:m2-det} lies between the bounds in \Cref{cor:m2-asymptotic} and is strictly larger than the deterministic last-signal value from \Cref{thm:adversial-det}.}
    \label{fig:appendix-adversarial-full-history-m2}
\end{figure}

\paragraph{Full history versus the last signal for $m=2$.}
The final adversarial experiment considers the multiple-signal model from \Cref{sec:multiple-vs-alpha}.  Revealing only the last signal $L=\max_j S_j$ is equivalent to a single $\alpha$-power signal with $\alpha=m$.  Thus, for $m=2$, the deterministic last-signal baseline is given by \Cref{thm:adversial-det} with $\alpha=2$, while the randomized last-signal value is given by \Cref{thm:adversarial-rand-opt}; \Cref{lem:last-signal-rand} shows that randomized algorithms gain no additional worst-case power from observing the full history.  For deterministic algorithms, however, full history can help.  We compute the exact full-history value for $m=2$ using \Cref{thm:m2-det} and compare it with the last-signal baselines.  The shaded region in \Cref{fig:appendix-adversarial-full-history-m2} shows the finite-$n$ bounds from \Cref{cor:m2-asymptotic}.  The figure illustrates the separation: the deterministic full-history value scales like $\frac3n$, whereas the deterministic last-signal value scales like $\frac2n$.

\end{document}